\documentclass[10pt]{article}
 
\usepackage{fullpage}
\usepackage[all=normal,bibliography=tight]{savetrees}
\usepackage{microtype}
\usepackage{todonotes}
\usepackage{amsthm}
\usepackage{amssymb,bbm}
\usepackage{microtype}

\usepackage{booktabs}

\usepackage{xspace}
\usepackage{todonotes}
\usepackage{comment}
\usepackage{afterpage}
\usepackage{pdflscape}
\usepackage{tikz}
\usepackage{graphicx}
\usepackage[absolute]{textpos}
\usetikzlibrary{arrows, decorations.markings}
\usetikzlibrary{fit}					
\usetikzlibrary{backgrounds}

\usepackage{amsmath,amssymb,amsthm}
\usepackage{graphicx}
\usepackage{caption}
\usepackage{subcaption}
\usepackage{url}
\usepackage{mathtools}

\newcommand{\uselipics}{no}

\numberwithin{equation}{section}
\newtheorem{theorem}{Theorem}[section]
\newtheorem{lemma}[theorem]{Lemma}
\newtheorem{claim}[theorem]{Claim}

\newtheorem{proposition}[theorem]{Proposition}

\theoremstyle{definition}
\newtheorem{definition}[theorem]{Definition}

\newcommand{\chead}[1]{\multicolumn{1}{c}{#1}}
\usepackage{paralist}
\usepackage[shortlabels]{enumitem}

\usepackage{mathtools}
\usepackage{microtype}

\usepackage{multirow}
\usepackage{tabularx}
\usepackage{dcolumn}
\usepackage{colortbl}
\renewcommand{\chead}[1]{\multicolumn{1}{c}{#1}}
\newcommand{\NAsymbol}{\textemdash} %
\newcommand{\NA}{\multicolumn{1}{c}{\NAsymbol}}
\newcolumntype{d}[1]{D{.}{.}{#1}} 
\newcolumntype{A}{>{\columncolor[gray]{0.9}}d{1.3}}
\newcommand{\NAgr}{\multicolumn{1}{>{\columncolor[gray]{0.9}}c}{\NAsymbol}}
\newcommand{\mhead}[1]{\multicolumn{2}{@{}c@{}}{#1}}
\newcommand{\mrow}[1]{\multirow{4}{*}{#1}}
\newcommand{\mrowNA}{\hfil\mrow{\NAsymbol}\hfill}

\newcommand{\iflipics}[2]{\ifthenelse{\equal{\uselipics}{yes}}{#1}{#2}}

\newcommand{\onlyfull}[1]{\iflipics{}{#1}}

\def\cqedsymbol{\ifmmode$\lrcorner$\else{\unskip\nobreak\hfil
\penalty50\hskip1em\null\nobreak\hfil$\lrcorner$
\parfillskip=0pt\finalhyphendemerits=0\endgraf}\fi} 

\def\Nesetril{Ne\v{s}et\v{r}il\xspace}

\def\Sanchez{S\'{a}nchez\xspace}

\def\Hlineny{Hlin{\v{e}}n{\'y}\xspace}
\def\Obdrzalek{Obdr{\v{z}}{\'a}lek\xspace}
\def\Gajarsky{Gajarsk{\'y}\xspace}

\def\maxoutdeg{\Delta^{\!+}}
\def\wcol{\operatorname{wcol}}

\newcommand{\widthm}[1]{ \mathrm{#1} } 
\DeclareMathOperator{\tw}{ \widthm{tw} }

\DeclareMathOperator{\td}{ \widthm{td} }

\newcommand{\cqed}{\renewcommand{\qed}{\cqedsymbol}}

\newcommand{\order}{L}
\usepackage{mathrsfs}

\newcommand{\mc}{\mathcal}
\newcommand{\Oh}{\mc{O}}

\newcommand{\Cc}{\mathscr{C}}

\newcommand{\N}{\mathbb{N}}
\newcommand{\R}{\mathbb{R}}

\newcommand{\col}{\mathrm{col}}

\newcommand{\WReach}{\mathrm{WReach}}
\newcommand{\SReach}{\mathrm{SReach}}

\def\grad_#1{\nabla\!_{#1}}
\def\topgrad_#1{\widetilde \nabla\!_#1}

\newcommand{\bd}{\mathrm{bd}}

\newcommand{\tree}{\texttt{tree1}\xspace}
\newcommand{\treeshrink}{\texttt{tree2}\xspace}
\newcommand{\ldit}{\texttt{ld\_it}\xspace}
\newcommand{\mfcs}{\texttt{mfcs}\xspace}
\newcommand{\tgva}{\texttt{new1}\xspace}
\newcommand{\tgvb}{\texttt{new2}\xspace}
\newcommand{\tgvc}{\texttt{new\_ld}\xspace}
\newcommand{\ldpow}{\texttt{ld}\xspace}

\makeatletter
\newcommand{\manuallabel}[2]{\def\@currentlabel{#2}\label{#1}}
\makeatother

\newtheorem*{rrule*}{Reduction Rule A1}
\newtheorem*{rrrule*}{Reduction Rule B1}

\title{Empirical Evaluation
of Approximation Algorithms for Generalized Graph Coloring and Uniform Quasi-Wideness}

\author{Wojciech Nadara\thanks{%
Institute of Informatics, University of Warsaw, Poland.
\texttt{wn341489@students.mimuw.edu.pl}
Supported by the ``Recent trends in kernelization: theory and experimental evaluation'' project, carried out within the Homing programme of the Foundation for Polish Science co-financed by the European Union under the European Regional Development Fund.}
\and
  Marcin Pilipczuk\thanks{%
Institute of Informatics, University of Warsaw, Poland.
\texttt{malcin@mimuw.edu.pl}
Supported by the ``Recent trends in kernelization: theory and experimental evaluation'' project, carried out within the Homing programme of the Foundation for Polish Science co-financed by the European Union under the European Regional Development Fund.}
\and
Roman Rabinovich\thanks{%
Lehrstuhl f\"ur Logic und Semantik, Technische Universit\"at Berlin, Berlin, Germany.
\texttt{roman.rabinovich@tu-berlin.de}}
\and
Felix Reidl\thanks{%
Department of Computer Science, Royal Holloway Universiy of London, London, United Kingdom. \texttt{felix.reidl@rhul.ac.uk}}
\and
Sebastian Siebertz\thanks{%
Institute of Informatics, University of Warsaw, Poland.
\texttt{siebertz@mimuw.edu.pl}.
The work of Sebastian Siebertz is supported by the National Science Centre of Poland via POLONEZ grant agreement UMO-2015/19/P/ST6/03998, 
which has received funding from the European Union's Horizon 2020 research and 
innovation programme (Marie Sk\l odowska-Curie grant agreement No.\ 665778).}}

\date{}
 
\begin{document}

\maketitle

\begin{textblock}{5}(12.77, 13.7)
\includegraphics[width=38px]{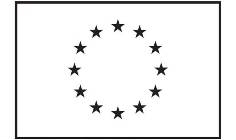}%
\end{textblock}

\begin{abstract}
\noindent The notions of \emph{bounded expansion} and \emph{nowhere denseness}
not only offer robust and general definitions of uniform sparseness of graphs,
they also describe the tractability boundary for several important algorithmic questions.
In this paper we study two structural properties of these graph classes that
are of particular importance in this context, namely the property of having
bounded \emph{generalized coloring numbers} and the property of being 
\emph{uniformly quasi-wide}.
We provide experimental evaluations of several algorithms that approximate
these parameters on real-world graphs.
On the theoretical side, we provide a new algorithm for uniform quasi-wideness
with polynomial size guarantees in graph classes of bounded expansion
and show a lower bound indicating that the guarantees of this algorithm
are close to optimal in graph classes with fixed excluded minor.
\end{abstract}

\section{Introduction}

\subsection{Sparse graph classes}

\paragraph{Treewidth and graph minors.} 
The exploitation of structural properties found in sparse graphs has a long
and fruitful history in the design of efficient algorithms. Besides the long
list of results on planar graphs and graphs of bounded degree (which are too
numerous to be fairly represented here), the celebrated structure theory of graphs
with excluded minors, developed by Robertson and Seymour~\cite{GM} falls into
this category. It not only had an immense influence on the design of efficient
algorithms (see e.g.\ \cite{DemaineH08,DemaineHK09}), it also introduced the
now widely used notion of treewidth (see e.g.~\cite{bodlaender1997treewidth})
and gave rise to the field of parameterized complexity.

To motivate our study of bounded expansion and nowhere dense classes, 
let us first elaborate a bit on the concepts of treewidth and minors. 
A graph has bounded treewidth if it can be recursively decomposed
along small separators. More formally, a clique in a graph $G$ is a 
set of vertices that are pairwise adjacent in $G$. If two graphs $G$ 
and $H$ each contain cliques of equal size, a clique-sum of $G$ 
and $H$ is formed from their disjoint union by identifying pairs 
of vertices in these two cliques to form a single shared clique, 
and then possibly deleting some of the clique edges. A $k$-clique-sum 
is a clique-sum in which both cliques have at most $k$ vertices. 
A graph has treewidth at most $k$ if it can be 
obtained via repeated  $k$-clique-sums of graphs starting with graphs with at most 
$k+1$ vertices. Hence, graphs of bounded treewidth can be 
decomposed into simple pieces of size at most $k+1$ that are
connected in a well controlled manner. Furthermore, such decompositions
can be computed efficiently both in theory~\cite{Bodlaender96}
and in practice~\cite{Tamaki17}, and many problems that are hard to 
solve in general can be solved efficiently on graphs of bounded
treewidth by dynamic programming (see e.g. the corresponding chapter of~\cite{dabook}). 

A graph $H$ is a minor of another graph $G$ if $H$ can be obtained 
from a subgraph of $G$ by contracting some edges. Equivalently, 
$H$ is a minor of $G$ if we can find pairwise disjoint, connected 
subgraph~$H_v$ of $G$, one for each $v\in V(H)$, such that whenever
$\{u,v\}$ is an edge of $H$, then also two vertices of~$H_u$ and 
$H_v$ are connected in $G$. If $G$ does not have $H$ as a minor, 
then we say that $G$ is $H$-minor-free. The structure theorem 
of Robertson and Seymour now reads as follows. For every graph $H$ 
there exists a number $k$ such that every $H$-minor-free graph can 
be obtained via repeated $k$-clique-sums of graphs that are $k$-almost
embeddable into a surface into which $H$ does not embed. We do not 
want to formally define the notion of almost embeddability. The point
we want to make is that again we can recursively decompose 
$H$-minor-free graphs into simpler pieces along small separators, 
so that we can apply dynamic programming techniques, while on the simpler
pieces we can apply topological arguments. While there exist efficient
algorithms to compute such decompositions, and the obtained algorithms
are efficient in theory, the constants appearing in the decomposition 
theorem are enormous and the algorithms fall short of being applicable
in practice. 

\paragraph{Graph classes of bounded expansion and nowhere dense graph classes.}
A complete paradigm shift was initiated by Ne\v{s}et\v{r}il and Ossona de
Mendez with their foundational work and introduction of the notions of
\emph{bounded expansion}~\cite{NesetrilM08,NesetrilM08a,NesetrilM08b} and
\emph{nowhere denseness}~\cite{NesetrilM11a}. These graph classes extend and properly contain $H$-minor-free classes and many arguments
based on topology can be replaced by more  general, and surprisingly 
often much simpler, arguments based on \emph{density}. We refer to 
the textbook~\cite{sparsity} for extensive background on the theory of 
sparse graph classes. Both notions are defined via the concept of excluded
\emph{bounded-depth minors}. We say that a graph $H$ is a \emph{minor
at depth~$r$} of another graph $G$, if we demand in the above definition
of a general minor that every subgraph~$H_v$ representing a vertex 
$v\in V(G)$ has radius at most $r$. Hence, minors at depth $0$ correspond
to subgraphs, minors at depth $1$ are obtained by identifying stars
in~$G$ with vertices of~$H$, and so on. In the limit, minors at depth $n$
correspond to general minors. 

We now define classes of bounded expansion and nowhere dense classes
as follows. A class $\Cc$ of graphs has \emph{bounded expansion}
if there exists a function $f\colon \N\rightarrow\N$ such that 
for every radius $r\in\N$ the density of depth-$r$ minors in graphs 
from $\Cc$ is bounded by $f(r)$. Similarly, a class $\Cc$ of graphs
is \emph{nowhere dense} if there exists a function $t\colon 
\N\rightarrow\N$ such that for every radius $r\in\N$ the graphs
from $\Cc$ exclude the complete graph $K_{t(r)}$ as a depth-$r$
minor. 

\subsection{Local structures}

As both notions of bounded expansion and nowhere dense
are defined by local constraints, we cannot expect
to find global decomposition theorems as in the case of $H$-minor-free
graphs. 
That is, we should not expect that graphs from a graph class of bounded expansion or a nowhere dense graph class
admit such a rigorous structure as a tree decomposition of bounded width, as it is the case of bounded treewidth graphs,
or even such a decomposition as the aforementioned Robertson-Seymour structure theorem for graphs excluding a fixed minor. 
Instead, we need to resort to a more local structures that may not be as powerful or descriptive as the aforementioned decompositions, but still admitting
a surprising number of algorithmic applications. 

\paragraph{$p$-treewidth and $p$-treedepth colorings.}
In these local structures in sparse graph classes again the notion of treewidth, and in fact the more restrictive notion
of treedepth, play key roles for the decomposition of bounded expansion
and nowhere dense classes. For a number $p$, a 
\emph{$p$-treewidth coloring} of a graph $G$ is a vertex coloring $\lambda
\colon V(G)\rightarrow C$ so that the subgraph of 
$G$ induced by the combination of any $i\leq p$ colors has treewidth
at most $i$. Hence, a $p$-treewidth coloring of a graph $G$ can 
be understood as a decomposition of $V(G)$ into disjoint pieces, so that 
any subgraph induced by at most $i$ pieces is strongly
structured -- it has treewidth at most $i$. There is a long line of research 
on low treewidth colorings
in $H$-minor-free graphs. It is a classical observation, underlying the
famous Baker approximation approach, that if in a connected planar graph
$G$ we fix a vertex~$v$ and color all vertices according to the residue
of their distance from $v$ modulo $p+1$, then the obtained coloring
with $p+1$ colors has the property that the union of any $p$ color
classes induces a graph of treewidth $\mathcal{O}(p)$. As proved
by Demaine et al.~\cite{demaine2005algorithmic} 
and by DeVos et al.~\cite{devos2004excluding}, such colorings
with $p+1$ colors can be found for any $H$-minor-free class of graphs. 
Decompositions of this kind are central in the design of approximation
and parameterized algorithms in $H$-minor-free graph classes. We refer
to \cite{demaine2005algorithmic,demaine2011contraction,
demaine2010approximation,devos2004excluding} for a broader discussion. 

When Ne\v{s}et\v{r}il and Ossona de Mendez introduced classes of
bounded expansion~\cite{NesetrilM08a}, they observed that these 
classes can be characterized by the existence of $p$-treewidth colorings
that use only $f(p)$ colors for some function $f$. They proved that 
in fact one can find even stronger colorings, namely, $p$-treedepth 
colorings. The treedepth of a graph is the minimum height of a rooted 
forest whose ancestor-descendant closure contains the graph; this 
parameter is never smaller than the treewidth. Analogously to $p$-treewidth
colorings, $p$-treedepth colorings are defined as colorings $\lambda\colon
V(G)\rightarrow C$ for some color set $C$ such that for all $i\leq p$
the combination of at most $i$ color classes induces a subgraph of
treedepth at most $i$. They proved that a class $\Cc$ of graphs has bounded
expansion if and only if there exists a function $f\colon \N\rightarrow\N$
such that for every $p\in \N$, every $G\in\Cc$ admits a $p$-treedepth
coloring with $f(p)$ colors. Similarly, a class of graphs $\Cc$ 
is nowhere dense if and only if there exists a function $f\colon 
\N\times \R\rightarrow\N$ such that for every $p\in \N$ and 
every real $\epsilon>0$, every $n$-vertex subgraph of a 
graph from $G\in\Cc$ admits 
a $p$-treedepth coloring with $f(p,\epsilon)\cdot n^\epsilon$ colors. 
A direct algorithmic application is a fixed-parameter algorithm for 
the subgraph isomorphism problem~\cite{NesetrilM08b, pilipczuk2018polynomial}. The algorithmic task is to find in 
a graph $G$ a pattern graph $H$ with~$p$ vertices. By finding a 
$p$-treedepth coloring of $G$ with a few colors first, 
the problem reduces to finding $H$ in a graph of treedepth $p$, 
which can be done in linear time.

The strong structure imposed by a $p$-treedepth coloring makes it often a tool of choice for
theoretical development of algorithms in graphs of bounded expansion.
Obviously, the number of colors
needed for a $p$-treedepth coloring leads to the dominating factor 
in algorithms based on this decomposition paradigm. Attempts to
get colorings with weaker structural properties but less colors
were made e.g.~in~\cite{KunOS18}.
Unfortunately, recent experiments by O'Brien and Sullivan~\cite{OBrienSullivan}
indicate that the number of colors in a $p$-treedepth coloring in real-world graphs
is too large for practical applications of this concept.

\paragraph{Generalized coloring numbers.}
This motivates to look at other, often weaker or at least less intuitive structures
that can be found in sparse graph classes and that witness their sparsity.
Low treedepth colorings are strongly related to the \emph{generalized 
coloring numbers}, which are one of the objects of study of the
present paper. The name ``coloring numbers'' may be misleading, 
as we are not coloring the vertices of a graph, but rather ordering them, 
but maybe the following analogy is sufficiently motivating the name. 
A graph $G$ is called $d$-degenerate if its vertices can be ordered
so that every vertex has at most $d$ smaller neighbors. Now, by a 
simple greedy procedure one can find a proper coloring of the 
vertices of $G$. Starting with the smallest vertex one colors the
vertices in increasing order. As every vertex has at most $d$
smaller neighbors in the order, one can always
find a color among $d+1$ colors that is not conflicting with 
the colors that these earlier colored neighbors received before. 
Hence, $d+1$ colors suffice to 
color the whole graph. This gives for example a simple procedure
to find a $6$-coloring of a planar graph, as planar graphs are 
$5$-degenerate. Therefore, the degeneracy of a graph is sometimes
called its coloring number (do not confuse this with its chromatic
number, which may be much smaller). 

The generalized coloring numbers can be seen as generalizations 
of the degeneracy order. They are vertex
orderings that, parameterized by a radius $r$, measure reachability 
properties at distance~$r$. In this work, we focus on one (arguably most popular and applicable)
  generalized coloring number called the \emph{weak coloring number} $\wcol_r$. 
For the exact definition we refer to Section~\ref{sec:wcol};
below we state an algorithmic and engineering goal we want to pursue.

\begin{quote}
\textsc{Weak coloring number} \\
\textbf{Input}: undirected graph $G$, radius $r$. \\
\textbf{Output}: an ordering $L$ of $V(G)$ with small weak coloring number $\wcol_r(G, L)$.
\end{quote}

A greedy coloring along such an order $L$
for an appropriate radius $r$ results in a $p$-treedepth coloring with
a number of colors depending on the quality of the order, 
as observed by Zhu~\cite{zhu2009colouring}. Therefore, the 
generalized coloring numbers provide an efficient way to approximate
$p$-treedepth colorings. Furthermore, they have also
direct combinatorial and algorithmic applications, including study
of the VC-density and neighborhood complexity of graphs~\cite{EickmeyerGKKPRS17,Pil2017number,ReidlVS16}, approximation 
and kernelization of distance-$r$ dominating sets~\cite{AmiriMRS17,DrangeDFKLPPRVS16,Dvorak13,dvovrak2017distance},
and construction of sparse neighborhood covers~\cite{GroheKS17}. 
Recall that a kernelization algorithm is a polynomial pre-processing 
algorithm that attempts to reduce the problem size up to the
point where a brute force algorithm leads to fixed-parameter tractability. 
Polynomial time pre-processing is an essential step for practical 
algorithms for hard combinatorial problems.
Also, while not expressed explicitly in these terms, the 
enumeration algorithm for first-order queries on sparse structured 
databases~\cite{KazanaS13} is essentially based on generalized 
coloring orderings. 

\paragraph{Uniform quasi-wideness.}
We now turn to combinatorial and algorithmic properties of nowhere 
dense graph classes. 
These classes are even more general than bounded expansion classes. 
These classes also admit low treedepth colorings and generalized coloring
orderings with few colors, however, there is an unavoidable factor
$n^\epsilon$ depending on the graph size in the bound on the number
of colors.
This greatly limits the algorithmic usability of these structures, as one
no longer can depend exponentially on the number of colors in a theoretical
running time bound of the algorithm in question (which is often the case
in graphs of bounded expansion). 

Therefore, other properties of these classes are more relevant
for practical algorithmic applications. One of these properties is 
\emph{uniform quasi-wideness}, which is a concept that was originally
studied in model theory~\cite{podewski1978stable} and finite
model theory~\cite{Dawar10}. A class $\Cc$ of 
graphs is \emph{wide} if for every radius $r$ and for every number
$m$, in every sufficiently large graph $G\in \Cc$ one can find a set
of $m$ vertices that are pairwise at distance greater than $r$. 
This is a very restrictive concept; not even the class of all stars
possesses it. However, it is instructive to consider the example of 
star graphs here. It may be possible to remove only a few vertices from
the graphs under consideration, in case of star graphs the centers
of the stars, to obtain a wide class. Exactly this is formalized in the 
definition of uniform quasi-wideness. A class $\Cc$ of graphs is
uniformly quasi-wide if for every radius $r$ there exists a number $s$
such that for every number~$m$ for sufficiently large graphs $G\in\Cc$
we can find a set of at most~$s$ vertices that can be removed 
from $G$ so that we find $m$ vertices at mutual distance greater
than $r$ in the resulting subgraph of $G$. 

Below we state an algorithmic and engineering goal we want to pursue:
\begin{quote}
\textsc{Uniform Quasi-Wideness} \\
\textbf{Input}: undirected graph $G$, set $A \subseteq V(G)$, radius $r$. \\
\textbf{Output}: as small as possible $S \subseteq V(G)$ and as large as possible
$B \subseteq A \setminus S$ such that the elements of $B$ are pairwise within distance
larger than $r$ in the graph $G-S$.
\end{quote}

Uniform quasi-wideness exactly characterizes nowhere dense graph classes:
a graph class closed under taking subgraphs is nowhere dense if and only
if it is uniformly quasi-wide~\cite{NesetrilM10}. 
Note that the uniform quasi-wideness definition does not impose an $n^\epsilon$
in any of the bounds and therefore is often the tool of choice for algorithms
in nowhere dense graph classes. 
On the other hand, it clearly does not give so intuitive and clear structural characterization
such as a tree decomposition or a $p$-treedepth coloring; to use it, one requires to find
a good leverage for it.

Intuitively, uniform quasi-wideness is a very useful property when
dealing with \emph{local properties} of graphs. This concept
was applied very successfully in parameterized complexity, 
e.g.\ to show that the distance-$r$ dominating set problem 
is fixed-parameter tractable on nowhere dense graph 
classes~\cite{DawarK09}, and in fact, more generally, testing
first-order properties is fixed-parameter tractable on nowhere
dense graph classes~\cite{GroheKS17}. The dominating set 
problem plays a central role in parameterized complexity as it
is the foremost example of a $\mathsf{W}[2]$-complete
problem. In fact, under the standard assumption that $\mathsf{FPT}
\neq \mathsf{W}[2]$, for subgraph closed classes, 
nowhere dense classes constitute 
the limit of algorithmic tractability for distance-$r$ dominating
set, distance-$r$ independent set and first-order 
model-checking~\cite{DrangeDFKLPPRVS16,DvorakKT13,pilipczuk2018kernelization}.
On the other hand, more and more sophisticated kernelization 
algorithms for distance-$r$ dominating set on nowhere dense
classes, which are all using the notion of uniform quasi-wideness, 
were developed~\cite{DawarK09,DrangeDFKLPPRVS16,
EickmeyerGKKPRS17,KreutzerRS17}. The concept was also 
applied in the context of lossy kernelization~\cite{EibenKMPS18}
and for efficient algorithms for the reconfiguration variants
of the above problems~\cite{lokshtanov2018reconfiguration,Siebertz17}. 

\subsection{Our contribution}
In summary, one core strength of the bounded expansion/nowhere 
dense framework is that there exists a multitude of equivalent 
definitions that provide 
complementing perspectives. We outlined two structural properties of
these classes that are of particular importance in the algorithmic context, 
namely the property of having bounded generalized coloring numbers 
and the property of being uniformly quasi-wide. 
Recall that probably the strongest and most intuitive one, $p$-treedepth colorings,
have been experimentally studied by O'Brien and Sullivan~\cite{OBrienSullivan}
with rather discouraging conclusions.

The central question of our work here is to investigate 
the two other outlined local structures. That is, we investigate
how the generalized coloring numbers and uniform quasi-wideness behave on
real-world graphs, an endeavor which so far has only been conducted for a 
single notion of bounded expansion and on a smaller scale~\cite{ComplexNetworks}.
Controllable numbers would be a prerequisite for practical
implementations of these algorithms based on such structural approaches.

\paragraph{Comparison of different approaches.}
We provide an experimental evaluation of several algorithms that approximate
these parameters on real world graphs. Our main goal is to identify which of
the approaches from the literature give best results and how they compare with simple
heuristics. 
That is, we do not provide here any start-to-end pipeline for any concrete optimization problem,
but rather aim at identifying the correct tools and algorithmic primitives for future applications.
We remark that a subsequent work of the first author~\cite{NadaraDomSet} uses the best
implementation for the uniform quasi-wideness property in an experimental study of
kernelization algorithms for \textsc{Dominating Set}. 

We describe the studied approaches for generalized coloring numbers in Section~\ref{sec:wcol}
and discuss the results of the experiments in Section~\ref{sec:wcol-results}.
The main finding is that all approaches with theoretical guarantees are outperformed by the simplest heuristic that sorts the vertices by their degrees. Note that this heuristic 
can be easily fooled by an artificial example. 
This simplest heuristic is in turn outperformed by two greedy approaches
that construct orderings from left-to-right or from right-to-left, making locally optimal
decisions.
Furthermore, all studied approaches benefit from a subsequent post-processing by a simple local
search routine that improves the quality of the ordering by at least a few per cent.

Similarly, the studied approaches to uniform quasi-wideness are described in Section~\ref{sec:uqw}
and the experimental results are presented and discussed in Section~\ref{sec:uqw-results}.
The comparison of approaches with theoretical guarantees reveal that the approach based
on so-called distance trees~\cite{Pil2017number} is superior to other methods.
However, we also find out that 
a very simple heuristic that deletes a few vertices of highest degree and
then computes the desired scattered set greedily outperforms all sophisticated approaches. 

\paragraph{Bounds on generalized coloring numbers on a large corpus of graphs.}
As a side result, the experiments yield bounds on weak coloring numbers
for a quite large corpus of real-world graphs
from different sources. 
We do not see any clear and rigorous
method of deciding whether these numbers are relatively small 
or large, that is, whether the studied graphs really come from some sparse graph class
with good bounds on the sparsity constants.
As a proxy, in Section~\ref{sec:statistics_wcol}
we discuss correlation between the obtained upper bounds for weak coloring numbers
and the graph size. Here, the main finding is that for radii $r \leq 3$ the weak
coloring numbers grow very slowly with the number of vertices of the graph
(which is expected in graphs of bounded expansion), but this breaks down for larger radii
(which is also somewhat expected as the radius approaches the logarithm of the number
 of vertices). 

We remark that the obtained numbers are only upper bounds on the weak coloring
numbers of the studied graph corpus, and we do not really know their exact values.
All known exact algorithms for computing the exact value of the weak coloring number
have exponential dependency on the graph size, which is infeasible even on our dataset of small
graphs (where graphs have around 200 vertices on average). 
While it is plausible that an involved branching algorithm with pruning is able to 
compute the exact value for this small dataset, developing such an algorithm and its implementation
seems challenging and beyond this work. 
Furthermore, we are not aware (and were not able to develop on our own) any good methods
of lower bounding the weak coloring numbers in a graph by, say,
exhibiting some small dense structure in a graph (in the same way as a large well-linked set
or a bramble of high order lower bounds the treewidth of a graph). 

Thus, being able to discover the \emph{exact} value of the weak coloring number of graphs
even for our dataset of small graphs remains a challenging future direction for research.
However, judging from the fact that on datasets of small- and medium-sized graphs various
approaches resulted in similar values of the weak coloring number, we guess that our values are not
far from the optimal ones. 

\paragraph{Contributions to the theory.}
Setting up the experiments led also to some contributions to the theory.
One of the studied approaches, combining generalized coloring numbers
with uniform quasi-wideness~\cite{KreutzerPRS16}, turned out to be very conservative 
in its choices. 
Inspired by the approach of~\cite{KreutzerPRS16},
we design a new algorithm for uniform quasi-wideness that avoids the conservative
steps and is arguably simpler.
In particular, our algorithm gives polynomial size guarantees in graph classes of bounded expansion.
Furthermore, we show a lower bound indicating that the guarantees of this algorithm
are close to optimal in graph classes with a fixed excluded minor.

\onlyfull{\medskip \noindent\textbf{Organization.} We give background on the theory of bounded
expansion and nowhere dense graphs in Section~\ref{sec:prelims}. In 
Section~\ref{sec:wcol} and Section~\ref{sec:uqw} we describe our approaches to compute 
the weak coloring numbers and uniform quasi-wideness. Our experimental
setup is described in Section~\ref{sec:setup} and our results are presented
in Section~\ref{sec:wcol-results} and Section~\ref{sec:uqw-results}.
Finally, Section~\ref{sec:lb} describes the lower bound for the new algorithm for
uniform quasi-wideness.

\hfill}

\section{Preliminaries}\label{sec:prelims}
\textbf{Graphs.}
All graphs in this paper are finite, undirected and simple, that is, they
do not have loops or multiple edges between the same pair of vertices. For
a graph $G$, we denote by $V(G)$ the vertex set of $G$ and by $E(G)$ its
edge set. If $U\subseteq V(G)$, then $G[U]$ means the subgraph of $G$ induced by $U$.
The \emph{distance} between a vertex $v$ and a vertex $w$ is the length
(that is, the number of edges) of a shortest path between $v$ and $w$. For
a vertex $v$ of $G$, we write \smash{${N^G(v)}$} for the set of all neighbors of~$v$, \smash{$N^G(v)=\{\,u\in V(G)\mid \{u,v\}\in E(G)\,\}$}, and for \smash{$r\in\N$} we
denote by \smash{$N_r^G[v]$} the \emph{closed $r$-neighborhood of~$v$}, that is,
the set of vertices of $G$ at distance at most~$r$ from $v$. Note that we
always have \smash{$v\in N^G_r[v]$}.
\iflipics{}{When no confusion can arise regarding the
graph $G$ we are considering, we usually omit the superscript~$G$.}
The radius of a connected graph~$G$ is the minimum integer $r$ such that
there exists $v\in V(G)$ with the property that all vertices of~$G$ have distance
at most $r$ to $v$. 
A set $A$ is \emph{$r$-independent} if all distinct vertices of $A$ have
distance greater than $r$. 

\smallskip\noindent
\textbf{Bounded expansion and nowhere denseness.} A {\em{minor model}} of 
a graph $H$ in a graph~$G$ is a family $(I_u)_{u\in V(H)}$ of pairwise 
vertex-disjoint connected subgraphs of $G$, called {\em{branch sets}},
such that whenever $uv$ is an edge in~$H$, there are $u'\in V(I_u)$ and 
$v'\in V(I_v)$ for which $u'v'$ is an edge in~$G$. The graph $H$ is a 
{\em{depth-$r$ minor}} of $G$, denoted $H\preccurlyeq_rG$, if there is a minor 
model $(I_u)_{u\in V(H)}$ of~$H$ in~$G$ such that each $I_u$ has radius at 
most $r$.

A \emph{topological minor model} of a graph $H$ in a graph~$G$ consists of an injective function
$f\colon V(H) \to V(G)$ and a family of paths $(P_{uv})_{uv \in E(H)}$. The path $P_{uv}$ connects $f(u)$ with $f(v)$ in $G$.
Furthermore, no other vertex from the image of $f$ lies on $P_{uv}$ and the paths $P_{uv}$ are pairwise vertex-disjoint except for the endpoints.
The graph $G$ is a \emph{depth-$r$ topological minor} of $G$ if there is a topological minor model of $H$ in $G$ with every path $P_{uv}$
of length at most $2r+1$.

A class $\Cc$ of graphs is \emph{nowhere dense} if there is a function 
$t\colon \N\rightarrow \N$ such that for all $r\in \N$ it holds that 
$K_{t(r)}\not\preccurlyeq_r G$
for all $G\in \Cc$, where $K_{t(r)}$ denotes the clique on $t(r)$ vertices.
A class~$\Cc$ has \emph{bounded expansion}
if there is a function $d\colon\N\rightarrow\N$ such that for all 
$r\in \N$ and all $H\preccurlyeq_rG$ with $G\in\Cc$, the {\em{edge density}}
of $H$, i.e.\ $|E(H)|/|V(H)|$, is bounded by $d(r)$. 
\iflipics{}{Note that every 
class of bounded expansion is nowhere dense. The converse is not necessarily true in general~\cite{sparsity}.}

\section{The weak coloring numbers}\label{sec:wcol}

\subsection{Definitions}
The \emph{coloring number $\col(G)$} of a graph $G$ is the minimum integer
$k$ such that there is a linear order~$L$ of the vertices of $G$
for which each vertex $v$ has \emph{back-degree} at most $k-1$, i.e., at most
$k-1$ neighbors $u$ with $u<_Lv$. It is well-known that for any graph $G$,
the chromatic number~$\chi(G)$ satisfies $\chi(G)\le \col(G)$, which possibly
explains the name ``coloring number''. 

We study a generalization of the coloring number that was introduced by 
Kierstead and Yang~\cite{kierstead03} in the context of coloring games 
and marking games on graphs. The \emph{weak coloring numbers} $\wcol_r$
are a series of numbers, parameterized by a positive integer $r$, which denotes
the radius of the considered ordering. 

The invariants $\wcol_r$ are defined in a way similar to the 
definition of the coloring number. Let $\Pi(G)$ be the set of all 
linear orders of the vertices of the graph
$G$, and let $L\in\Pi(G)$. Let $u,v\in V(G)$. For a positive integer $r$, 
we say that $u$ is \emph{weakly $r$-reachable} from~$v$ with respect to~$L$, 
if there exists a
path $P$ of length~$\ell$, $0\le\ell\le r$, between $u$ and $v$ such that
$u$ is minimum among the vertices of $P$ (with respect to $L$). Let
$\WReach_r[G,L,v]$ be the set of vertices that are weakly $r$-reachable
from~$v$ with respect to $L$. Note that $v\in\WReach_r[G,L,v]$. 
The \emph{weak $r$-coloring number
  $\wcol_r(G)$} of $G$ is defined as
\begin{equation*}
  \wcol_r(G)\coloneqq \min_{L\in\Pi(G)}\:\max_{v\in V(G)}\:
  \bigl|\WReach_r[G,L,v]\bigr|\,.
\end{equation*}

As proved by Zhu \cite{zhu2009colouring}, 
the weak coloring numbers can be used to characterize bounded
expansion and nowhere dense classes of graphs: A class $\Cc$ of graphs has
bounded expansion if and only if there exists a function $f\colon\N\rightarrow\N$
such that $\wcol_r(G)\leq f(r)$ for all $r\in\N$ and all $G\in\Cc$. A class~$\Cc$ is nowhere dense if and only if there is a function $f\colon \N\times\R\rightarrow \N$ such that for every real $\epsilon>0$ 
and every $r\in \N$ and all $n$-vertex graphs $H$ 
that are subgraphs of some $G\in\Cc$ we have $\wcol_r(H)\leq f(r,\epsilon)\cdot n^\epsilon$. 

An interesting aspect of the weak coloring numbers is that these
invariants can also be seen as gradations between the coloring number
$\col(G)$ and the \emph{treedepth $\td(G)$} (which is the
minimum height of a depth-first search tree for a supergraph of $G$
\cite{nevsetvril2006tree}). 
More explicitly, for every graph $G$ we have (see \cite[Lemma~6.5]{sparsity})
\[\col(G)= \wcol_1(G)\le \wcol_2(G)\le
  \dots\le \wcol_\infty(G)= \td(G)\,.\]
Consequently, we also consider an algorithm for computing treedepth in 
our empirical evaluation.

A related notion to weak coloring numbers are \emph{strong} coloring
numbers, which were also introduced in~\cite{kierstead03}. 
Let $L\in \Pi(G)$, let $r$ be a positive integer and let $v\in
V(G)$. We say that a vertex~$u$ is \emph{strongly $r$-reachable} from $v$ if
there is a path~$P$ of length $\ell$, $0\le \ell \le r$ such that $u=v$ or $u$ is
the only  vertex of $P$ smaller than~$v$ (with respect to~$L$). Let
$\SReach_r[G,L,v]$ be the set of vertices that are strongly $r$-reachable from
$v$ with respect to $L$. Again, $v\in \SReach_r[G,L,v]$. The \emph{strong
  $r$-coloring number} $\col_r(G)$ is defined as $\col_r(G) \coloneqq \min_{L\in\Pi(G)}\:\max_{v\in V(G)}\:
  \bigl|\SReach_r[G,L,v]\bigr|.$ As weak coloring numbers converge to
  treedepth with growing $r$, strong coloring numbers converge to treewidth~\cite{GroheKRSS15}:
  \[\col(G)= \col_1(G)\le \col_2(G)\le
  \dots\le \col_\infty(G)= \tw(G)\,.\]
The reason is that treewidth of $G$ can be characterized by the minimal width of
an elimination ordering of $G$ defined exactly as $\col_\infty(G)$.

Clearly, for all $r\in\N$, $\col_r(G) \le \wcol_r(G)$ (and thus $\tw(G)\le
\td(G)$). Moreover, for all $r$ we have $\wcol_r(G)\le
(\col_r(G))^r$~\cite{kierstead03}. It follows that for every graph $G$ there is
some (possibly large) integer $r$ such that $\wcol_{r-1}(G) \le \tw(G) \le \wcol_r(G)$. This gives a hope that an
elimination ordering computed for treewidth gives a good upper bound for
$\wcol_{r'}(G)$ where $r'\le r-1$. We we will evaluate orders produced by an algorithm for
treewidth approximations, but interpreted as an order for weak coloring numbers.

Concrete bounds for the weak coloring numbers on restricted graph classes
are given in~\cite{GroheKRSS15, KreutzerPRS16,NesetrilM08,FelixThesis,HeuvelMQRS17,zhu2009colouring}. 
The approximation algorithms we study are based on the approaches described in 
\cite{NesetrilM08,FelixThesis, HeuvelMQRS17}, which we describe in more
detail in the following subsections.


\subsection{Distance-constrained Transitive Fraternal Augmentations}  
\def\dir#1{\vec{#1}}

\iflipics{\input{dtf-short}}{In this section we describe an approach based on distance-constrained transitive fraternal augmentations, developed in~\cite{NesetrilM08a,FelixThesis}.
In~\cite{FelixThesis} one can find the following guarantee.
\begin{theorem}[\cite{FelixThesis}]
Given a graph $G$ and an integer $r$, one can construct an ordering $L$ of $V(G)$ with the sizes of weakly reachable sets bounded by
$$2^{2^{\Oh(r)}} \cdot (\widetilde{\bigtriangledown}_{r+1}(G) \cdot \bigtriangledown_0(G))^{2^{\Oh(r)}},$$
where $\widetilde{\bigtriangledown}_{r+1}(G)$ is the maximum density of depth-$(r+1)$ topological minors in $G$
while $\bigtriangledown_0(G)$ is the maximum density of depth-0 minors (i.e., subgraphs) of $G$ (that is, the degeneracy of $G$). 
\end{theorem}

Given a graph $G$ and a linear order $L$ of its vertices, observe that
we have the following properties:
\begin{enumerate}
\item Let $u,v,w\in V(G)$ be such that $v\in \WReach_i[G,L,u]$ and 
$w\in \WReach_j[G,L,u]$ for some numbers $i,j$. 
Then either $v\in \WReach_{i+j}[G,L,w]$ or $w\in \WReach_{i+j}[G,L,v]$.
\item Let $u,v,w\in V(G)$ be such that $u\in \WReach_i[G,L,v]$ and 
$v\in \WReach_j[G,L,w]$ for some numbers $i,j$.
Then $u\in \WReach_{i+j}[G,L,w]$.
\end{enumerate}

\noindent
We can approximate the weak coloring numbers by
orienting the input graph $G$ and iteratively inserting arcs so that 
the above reachability properties are satisfied. Introducing
an arc with the aim of satisfying property~1 above is called 
a \emph{fraternal augmentation}, while introducing an arc
with the aim of satisfying property~2 is called a 
\emph{transitive augmentation}. These operations were 
studied first in~\cite{NesetrilM08a}. We are going to work with 
an optimized version, called \emph{distance-constrained
transitive-fraternal augmentations}, short \emph{dtf-augmentations}, 
which was introduced in~\cite{FelixThesis} as a more practical variant 
of transitive-fraternal augmentations. 

Let $G$ be an undirected graph and let $\dir G_1$ be any orientation of~$G$.
Then a dtf-augmentation of~$G$ is a sequence $\dir G_1 \subseteq \dir G_2 \subseteq \dots$ of
directed graphs which satisfy the following two constraints:
\index{dtf-augmentation}
\begin{enumerate}
    \item Let $u,v,w\in V(G)$ be such that $uv \in E(\dir G_i)$ and $uw \in E(\dir G_j)$
    are arcs of $\dir G_i$ and $\dir G_j$, respectively. Then it
    follows that either $vw \in E(\dir G_{i+j})$ or $wv \in E(\dir G_{i+j})$.
    \item Let $u,v,w\in V(G)$ be such that $vu \in E(\dir G_i)$ and $wv \in E(\dir G_j)$ are arcs of $\dir G_i$ and $\dir G_j$, respectively. Then it
    follows that $wu \in \dir G_{i+j}$. 
\end{enumerate}
Just as above, arcs added because of the first item are called \emph{fraternal}
and arcs added because of the second item are called \emph{transitive}.  To
simplify notation we associate a weight function
$\omega_i \colon V(\dir G_i)^2 \to \{1,\ldots,\}\cup\{\infty\}$ with the $i$-th
dtf-augmentation~$\dir G_i$ where $w_1(uv) = 1$ if $uv\in E(\dir G_1)$ and
$w_1(uv) = \infty$ if $uv\notin E(\dir G_1)$ and
\begin{align*}
  \omega_i(uv) = \begin{cases}
    \min\{ \omega_{i-1}(uv), i \}  &\text{if~$uv \in \dir G_i$} \\
    \infty                         &\text{else}
  \end{cases}
\end{align*}
In other words: if the arc~$uv$ is present in $\dir G_i$ but not in $\dir
G_{i-1}$, then we have $\omega_{\geq i}(uv) = i$ and $\omega_{< i}(uv) =
\infty$. It can be shown that the arcs of weight~$d$ appear exactly in
augmentation~$\dir G_d$. These augmentations behave similarly
to graph powers in the following sense: consider two vertices~$u,v$ that are
at distance~$d$ in~$G$. Then in every augmentation~$\dir G_r$ for~$r \geq d$ we
either find the arc~$uv \in \dir G_d$ with~$\omega_d(uv) = d$, or the
arc~$vu \in \dir G_d$ with~$\omega_d(vu) = d$, or we find a common out-neighbor~$w$ of~$u$ and~$v$ in~$\dir G_d$ such that~$\omega_r(wu) +
\omega_r(wv) = d$.

Importantly, graph classes of bounded expansion admit dtf-augmentations in
which the maximum out-degree~$\Delta^{\!+}(\dir G_r)$ depends only on a function of 
depth~$r$ and on the graph class in question~\cite{FelixThesis} (we remark that
commonly in the literature one orients the graphs $\dir G_i$ to minimize in-degrees
instead of out-degrees, however, for consistency with the weak coloring numbers
we orient so that an arc $uv\in E(\dir G_i)$ corresponds to $u\in \WReach_i[G,L,v]$). 
The algorithm to
compute such augmentations closely follows the original algorithm for tf-augmentations (described
in~\cite{NesetrilM08a,sparsity}): first, the orientation~$\dir G_1$ is chosen to be the
acyclic ordering derived from the degeneracy ordering of~$G$; this orientation
minimizes~$\maxoutdeg(\dir G_1)$. Second, we
can orient the fraternal arcs added in step~$r$ by first collecting \emph{all} potential
fraternal edges in an auxiliary graph~$G^f_r$ and then
again compute an acyclic orientation $\dir G^f_r$ which minimizes the out-degree. We then insert the arcs into~$\dir G_r$ according to their orientation
in~$\dir G^f_r$.

If instead of
computing fraternal edges at step~$r$ by searching for fraternal
configurations in all pairs $\dir G_i$, $\dir G_j$ with~$i+j = r$, it suffices
to consider the pair~$\dir G_{r-1}$, $\dir G_1$. The same optimization
does \emph{not} hold for transitive arcs, however.

The precise connection between dtf-augmentations and $\wcol$-orderings is
presented in the following lemma. 

\begin{lemma}[\cite{AmiriMRS17,GroheKS17}]
  Let~$\dir G_r$ be the $r$-th dtf-augmentation of a graph~$G$ and let~$G_r$ be
  the underlying undirected graph. Let~$L$ be an ordering of~$V(G)$ such that 
  every vertex has at most~$c$ smaller neighbors with respect to~$L$. Then
  $\operatorname{WReach}_r[G,L,v] \leq (\maxoutdeg(\dir G_r)+1)c + 1$ for all~$v \in G$.
\end{lemma}

Therefore we can obtain a $\wcol_r$-ordering from the $r$th dtf-augmentation~$\dir G_r$
by simply computing a degeneracy ordering of~$G_r$.

}

\subsection{Flat decompositions}\label{ss:flat}

\iflipics{\input{short-flat}}{\newcommand{\GHi}[1]{G[H_{\ge#1}]}
The following approach for approximating the weak coloring numbers 
was introduced in \cite{HeuvelMQRS17} and provably yields good results
on graphs that exclude a fixed minor. 

\begin{theorem}[\cite{HeuvelMQRS17}]
Let $r \geq 1$ and $t \geq 4$ be integers and assume that $G$ does not contain $K_t$ as a minor.
Then
$$\wcol_r(G) \leq \binom{r+t-2}{t-2} \cdot (t-3)(2r+1) \in \Oh(r^{t-1}).$$
\end{theorem}

A \emph{decomposition} of a graph~$G$ is a sequence
$\mathcal{H}=(H_1,\ldots,H_\ell)$ of non-empty subgraphs of $G$ such that
the vertex sets $V(H_1),\ldots,V(H_\ell)$ partition $V(G)$. The
decomposition $\mathcal H$ is \emph{connected} if each~$H_i$ is
connected. 

A decomposition of a graph $G$ induces a partial order on $V(G)$
by defining $u<v$ if $u\in V(H_i)$ and $v\in V(H_j)$ for $i<j$. A decomposition
yields a good order for the weak coloring numbers for a given $r$ if we can 
\begin{enumerate}[(1)]
\item guarantee that the $r$-neighborhood of each $v\in V(H_j)$
has a small intersection with $\bigcup_{i<j}V(H_i)$ (then, in particular,
$\WReach_r[G,L,v]\cap \bigcup_{i<j}V(H_i)$ is small), and 
\item ensure that we can 
order the vertices inside each $H_i$ so that we have good weak
reachability properties. 
\end{enumerate} 
We call such a decomposition \emph{flat.}
The following procedure was proposed in~\cite{HeuvelMQRS17} to compute a 
decomposition of a graph $G$. If $G$ excludes the complete graph
$K_t$ as a minor, the resulting decomposition is flat. For a decomposition $(H_1,\ldots,H_\ell)$ of a 
graph~$G$ and $1\le i\le\ell$, we denote by $\GHi{i}$ the subgraph 
of~$G$ induced by $\bigcup_{i\le j\le\ell}V(H_j)$.

  Without loss of generality we may assume that $G$ is connected. We
  iteratively construct a connected  decomposition
  $H_1,\ldots,H_\ell$ of $G$, see Figure \ref{fig:flat} for an example. 
  To start, we choose an arbitrary vertex $v\in V(G)$ and let $H_1$ be the
  connected subgraph $G[v]$. 
  Now assume that for some~$q$, $1\le q\le\ell-1$, the sequence
  $H_1,\ldots,H_q$ has already been constructed and let $G'$ be the graph induced by vertices not in $H_i$, i.e., $G' = G[V(G) \setminus \bigcup_{1 \le i \le q }V(H_i)]$. Fix some component $C$ of
  $G'$ and denote by 
  $Q_1,\ldots,Q_s\in \{H_1,\ldots, H_q\}$ the subgraphs that have a 
  connection to $C$.  Using that $K_t$ is excluded as a minor, 
  one may argue that $s\leq t-2$.
  Because~$G$ is connected, we have
  $s\geq 1$. Let $v$ be a vertex of $C$ and let 
  $T$ be a breadth-first search tree in $G[C]$
  with root~$v$. We choose $H_{q+1}$ to be a minimal connected subgraph of
  $T$ that contains $v$ and that contains for each~$i$, $1\le i\le s$, at
  least one neighbor of $Q_i$.  
  As shown in~\cite{HeuvelMQRS17}, if $K_t\not\preccurlyeq G$, then the above
procedure produces a linear order~$L$ that certifies that $\wcol_r(G)\in
O(r^{\,t-1})$. 

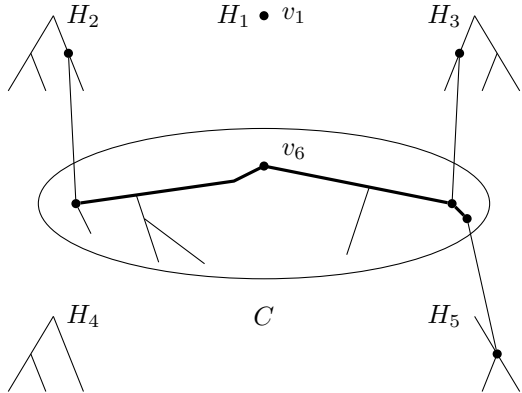
\begin{figure}
  \centering
  \begin{tikzpicture}
  \tikzstyle{vertex}=[circle,inner sep=1,minimum size =1mm,semithick,fill=black, draw=black]

  \node[vertex] (v1) at (0,0){};
  \node at (0.4,0) {$v_1$};
  \node at (-0.4,0) {$H_1$};
  \node at (-2.4,0) {$H_2$};
  \node at (2.4,0) {$H_3$};
  \node at (-2.4,-4) {$H_4$};
  \node at (2.4,-4) {$H_5$};

  \draw (2.8,0) -- (3.4,-1);
  \draw (2.8,0) -- (2.4,-1);
  \draw (3.1,-0.5) -- (2.9,-1);

  \draw (-2.8,0) -- (-3.4,-1);
  \draw (-2.8,0) -- (-2.4,-1);
  \draw (-3.1,-0.5) -- (-2.9,-1);

  \draw (-2.8,-4) -- (-3.4,-5);
  \draw (-2.8,-4) -- (-2.4,-5);
  \draw (-3.1,-4.5) -- (-2.9,-5);

  \draw (2.8,-4) -- (3.4,-5);
  \draw (3.1,-4.5) -- (2.9,-5);

  \draw (0,-2.5) ellipse (3cm and 1cm);
  \node at (0,-4){$C$};

  \node[vertex] (2) at (-2.6,-0.5){};
  \node[vertex] (3) at (2.6,-0.5){};
  \node[vertex] (5) at (3.1,-4.5){};

  \node[vertex] (c2) at (-2.5,-2.5){};
  \node[vertex] (c3) at (2.5,-2.5){};
  \node[vertex] (c5) at (2.7,-2.7){};

  \draw (2) -- (c2);
  \draw (3) -- (c3);
  \draw (5) -- (c5);
  
  \node[vertex] (v6) at (0,-2){};
  \node at (0.4,-1.8){$v_6$};

  \draw[very thick] (v6.center) -- ++(-0.4,-0.2) -- (c2);
  \draw (c2)-- ++(0.2,-0.4);
  \draw (v6.center)  ++(-1.7,-0.38) -- ++(0.3,-0.9);
  \draw (-1.59,-2.7) -- ++(0.8,-0.6);
  \draw[very thick] (v6.center) -- (c3) -- (c5);
  \draw (v6.center)  ++(1.4,-0.28) -- ++(-0.3,-0.9);
\end{tikzpicture}

\caption[Flat decomposition]{Construction of a flat decomposition. Thick lines in component $C$ are the new $H_6$.}
\label{fig:flat}
\end{figure}


\subsubsection{Implementation details}
Observe that this procedure leaves some freedom on how to pick the vertex
$v$ of $C$ from which we start the breadth-first search
and in which order to insert the vertices of $H_i$.
We evaluate several options. For the choice of the root vertex, 
   the following choices seem reasonable. 
\begin{enumerate}
\item Choose a vertex that is maximizing the number of neighbors
in some $Q_i$, to possibly
obtain a set $V(H_{q+1})$ that is smaller than when we choose a vertex far from
all $Q_i$.\label{flat:1}
\item Choose a vertex that has maximum degree in $C$, high degree vertices
should be low in the order. \label{flat:2}
\item Choose a vertex that has maximum degree in $C$, but only among those that
are adjacent to some $Q_i$.\label{flat:3}
\end{enumerate}
For the order of the vertices of $H_i$, we check the following options.
\begin{enumerate}
\item The breadth-first search and the depth-first search order from the root.
\item Sorted by degrees, non-increasingly.
\item Each of the above, but reversed.
\end{enumerate}
}

\subsection{Two known heuristics for a related graph parameter}

\subsubsection{Treedepth heuristic}

Since the `limit' of weak-coloring numbers is exactly the treedepth of a graph,
i.e., $\wcol_{\infty}(G) = \td(G)$, we consider simply computing a treedepth
decomposition and using an ordering derived from the decomposition. Our
algorithm of choice, developed by \Sanchez{} Villaamil~\cite{FernandoThesis} and
implemented by Oelschl\"agel~\cite{TDImplementation},%
\footnote{To the best of our knowledge, the cited thesis of Oelschl\"{a}gel~\cite{TDImplementation}
  is not available on the web. Our repository~\cite{our-repo} includes
    the source code by Oelschl\"{a}gel, while the thesis of
    \Sanchez{} Villaamil~\cite{FernandoThesis} describes the heuristic.}
recursively extracts
separators from the graph. To minimize the search space, only \emph{close}
separators are considered, that is, separators~$S$ that lie in the closed
neighborhood of some vertex. Furthermore, the algorithm makes use of the
following proposition.

\begin{proposition}[\cite{BerrySeparators}]
  If~$S \subseteq G$ is a minimal separator of a graph~$G$ and~$x \in S$, then
  for each connected component~$C$ of~$G - (S \cup N(x))$ the set~$N(C)$ is 
  a minimal separator of~$G$.
\end{proposition}

\noindent 
Let $N_S(G)$ be the set of minimal separators that can be constructed from a
minimal separator $S$ by applying the above proposition, where~$S$ is an 
arbitrary minimal close separator. The algorithm then finds the separator
$S_0 \in N_S(G)$ which minimizes the size of the largest connected component
in~$G - S_0$ (the implementation supports other heuristics, but this heuristic
turned out to have an acceptable running time for the large instances).

\subsubsection{Treewidth heuristic}\label{subsec:min-fill-in}

A well-known approach to compute a treewidth decomposition of a graph
is to find a linear order of the vertices, an elimination order, of
possibly small maximum back-degree. From such an order it is easy to
construct a tree decomposition of width equal to the back-degree (see,
e.g.\@ \cite{BodlaenderKoster2010}). Let
$L\in \Pi(G)$ and let $v\in V(G)$. The
\emph{back-degree} of $v$ is defined as
\[ \bd(v,G) \coloneqq |\SReach_\infty[G,L,v]|\,. \]
There are a number of heuristics to produce good elimination
orders. We chose one that is simple, fast
and that gives rather good results for treewidth: the so-called
minimum-degree heuristic~\cite{BodlaenderKoster2010}.

The minimum-degree algorithm orders the vertices of the
graphs starting from the biggest vertex which is one with minimum
degree. Assume that we already ordered vertices with indices greater
than~$i$, we put on position $i$ a vertex with the least back-degree.

\subsection{New heuristics}
\subsubsection{Greedy approach based on weakly reachable sets}\label{ss:WReachLeft}

Since our goal is to construct an ordering minimizing the largest weakly reachable set,
we propose the following greedy approach.

The crucial observation for our heuristic is that
the set $\WReach_r[G,L,v]$ depends on the partition of vertices of $V(G) \setminus \{v\}$
into vertices smaller and larger than $v$ in $L$, depends on the relative order in $L$
of vertices smaller than $v$, but \emph{does not} depend on the relative order of vertices
larger than $v$.
Furthermore, if in a given ordering $L$ one moves a vertex $v$ to a later position in the order,
  then the set $\WReach_r[G,L,v]$ can only increase.

This motivates the following approach. We compute an order $L$ from left to right.
Having already decided on a set $V' \subseteq V(G)$ as the smallest $|V'|$ vertices
in the constructed order $L$ and an ordering~$L'$ of $V'$, we compute for every $v \in V(G) \setminus V'$
the size of the weakly reachable set of $v$, assuming that~$v$ is the next vertex in the ordering. 
At every step, we take a vertex with the largest set, breaking ties by degrees (i.e., preferring vertices of larger degrees).

We optimize the running time of this greedy algorithm as follows. 
For every vertex $v \in V(G) \setminus V'$, we maintain its current weakly reachable set assuming that $v$ is placed next in the ordering, called henceforth \emph{potential weakly reachable set of $v$}.
Observe that, whenever we decide to place some vertex $v_0 \in V(G) \setminus V'$ as the next vertex in the constructed ordering,
it affects the potential weakly reachable sets of the remaining vertices only in the following fashion: some of them may additionally include now $v_0$. 
The set of vertices of $V(G) \setminus V'$ that now start to contain $v_0$ in their potential weakly reachable sets
can be discovered by a single depth-$r$ breadth first search from~$v_0$ in $G-V'$.
Observe that the number of vertices visited by all the breadth first searches in the algorithm equals the total size of all constructed weakly reachable sets, and thus we expect it to be much smaller than quadratic in $n$.

\subsubsection{Greedy approach based on strongly reachable sets}\label{ss:SReachRight}

We also propose a modification of the previous heuristic that constructs the order from right to left (i.e., from vertices later in the order to smaller).

If we decide to go from right to left, we cannot compute potential weakly reachable sets as previously, since $\WReach_r[G,L,v]$ depends on the relative order of vertices smaller than $v$.
Thus, we use a related notion of strongly reachable sets, $\SReach_r[G,L,v]$. 
Here, the crucial observation is that $\SReach_r[G,L,v]$ only depends on the partition of $V(G) \setminus \{v\}$ into vertices smaller and larger than $v$ in $L$.

For $i\in\{0,\ldots,|V(G)|-1\}$, assume that we have already decided to place $V_i \subseteq V(G)$ as the $i = |V_i|$ largest vertices (thus $V_0 = \emptyset$). For $i+1$,  we compute,
for every $v \in V(G) \setminus V_i$, the strongly reachable set of $v$ (called henceforth the \emph{potential strongly reachable set})
if $v$ is placed next in the order, i.e., $\SReach[G,L',v]$ for some $L'$ with $V_i$ as the largest $i$ vertices and $v$ the next largest one. Here we use that the result is the same for every such $L'$. We choose a vertex $v_i$ with the smallest potential strongly reachable set, breaking ties by degrees (i.e., preferring vertices of smaller degree) and define $V_{i+1} = V_i \cup \{v_i\}$.

We optimize the running time of this greedy algorithm as follows. 
For every vertex $v \in V(G) \setminus V'$, we maintain its potential strongly reachable set as a balanced binary search tree (\texttt{set} from the STL library in C++).
Assume that a vertex $v_0$ is placed as next in the ordering, and let $S$ be its potential strongly reachable set.
The crucial observation is that only potential strongly reachable sets of vertices from $S$ change: first, they lose $v_0$ and second, they may gain new vertices by paths passing through $v_0$.
The latter can be discovered as follows.
We partition $S$ into layers $S_1,S_2,\ldots,S_r$, where $S_i \subseteq S$ are vertices $v$ whose shortest path from $v_0$ to $v$ via $V'$ is of length exactly $i$.
After putting $v_0$ into the constructed order, the potential strongly reachable set of $v \in S_i$ starts to include the whole $S_j$ for every $i + j \leq r$.
Our algorithm computes layers $S_i$ by breadth-first search and then iterates over all choices of indices $i,j$ with $i+j \leq r$ and inserts every $w \in S_j$ into the potential strongly reachable set of every $v \in S_i$.

\subsubsection{Sorting by degrees and other simple heuristics}\label{ss:simple}

We also included in the comparison the following naive heuristics.
\begin{itemize}
\item
For $r=1$ an optimal order is a degeneracy order, which can be easily computed.
We can check if this order produces reasonable results for higher values of $r$ as well.
\item
Intuitively, it makes sense to sort vertices by descending degree (ties are broken arbitrarily)
because from vertices of high degree more vertices can be reached in one step.
This intuition is further supported by one popular network model, the \emph{Chung--Lu} random
graphs which sample graphs with a fixed degree distribution and successfully replicate several
statistics exhibited by real-world networks~\cite{chung2002average,chung2002connected}. In this
model, vertices are assigned weights (corresponding to their expected degree) and edges are sampled
independently but biased according to the endpoints weights. Under this model, vertices of the same
degree are exchangeable and the one ordering we can choose to minimize the number of $r$-reachable
vertices is simply the descending degree ordering. 
\item
A simple idea of generalizing the above heuristics to bigger values of $r$
is to apply them to the $r$th power $G^r$ of $G$, i.e., $G^r$ is defined as the graph
with
$V(G^r) = V(G)$ and $uv \in E(G^r) \Leftrightarrow \mathrm{dist}_G(u, v) \le r$.
\item As a baseline, we also included random ordering of vertices.
\end{itemize}

\subsection{Local search}
In addition to all these approaches we can try to improve their results
by local search
, a technique where we make small changes
to a candidate solution. 
%
We applied the following local changes and tested whether they caused
improvements to the current order $L$. 
\begin{itemize}
\item Take any vertex $v$ that has biggest $\WReach_r[G,L,v]$ and swap it with a random vertex that is smaller with respect to $L$.
\item Take any vertex $v$ that has biggest $\WReach_r[G,L,v]$ and swap it with its direct predecessor $u$ in $L$.
\end{itemize}

Both heuristics try to place a vertex with many weakly reachable vertices to the
left of them and thus to make them non-weakly reachable.
The advantage of the second rule is that the only possible changes are that
$\WReach_r[G,L,v]$ loses $u$ (if $u$ was there) and that $\WReach_r[G,L,u]$ may
obtain $v$. So $\WReach_r[G,L,v]$ is trivial to recompute and the only
computationally heavy update is for the new $\WReach_r[G,L,u]$.
For the first rule, recomputing $\WReach$ sets is more expensive.
However, the disadvantage of the second rule is that it does not lead to further
improvements quickly, hence applications of only the first rule
give better results than applications of the second rule only. 
In our implementation we did a few optimizations
in order to improve the results of the second rule, but we refrain from describing 
them in detail. The
final algorithm conducting local search firstly performs a round of applications of the first rule and when
they no longer improve the results, it performs a round of applications of the second rule. This combination turned out
to be empirically most effective.


\section{Uniform quasi-wideness}\label{sec:uqw}

Intuitively, a class of graphs is \emph{wide} if for every graph~$G$
from the class, every radius $r\in\N$ and every large subset $A\subseteq V(G)$ 
of vertices one can find a large subset $B\subseteq A$ of vertices 
which are pairwise at distance greater than $r$ (recall that such a 
subset is called \textit{$r$-independent}). The notion of 
uniform quasi-wideness allows to additionally delete a small number of
vertices to make $B$ $r$-independent. The following definition formalizes 
the meaning of ``large'' and ``small''. 

\begin{definition}
A class $\Cc$ of graphs is \emph{uniformly quasi-wide} if for every
$m\in \N$ and every $r\in \N$ there exist numbers $N(m,r)$ and $s(r)$ such 
that the following holds. 
\medskip
\begin{quotation}
\noindent \textit{Let $G\in\Cc$ and let $A\subseteq V(G)$ with $|A|\geq N(m,r)$. Then
there exists a set $S\subseteq V(G)$ with $|S|\leq s(r)$ and a set $B\subseteq 
A\setminus S$ of size at least $m$ 
such that for all distinct $u,v\in B$ we have $\mathrm{dist}_{G-S}(u,v)>r$.} 
\end{quotation}
\end{definition}

Uniform quasi-wideness was introduced by Dawar in~\cite{Dawar10}
and it was proved by Ne\v{s}et\v{r}il and Ossona de Mendez in~\cite{NesetrilM10} 
that uniform quasi-wideness is equivalent to nowhere denseness. Very recently, 
it was shown that the function $N$ in the above definition can be chosen to 
be polynomial in $m$~\cite{KreutzerRS17,Pil2017number}. A single exponential
dependency was earlier established for classes of bounded expansion~\cite{KreutzerPRS16}. We are going to evaluate the algorithms derived from
the proofs in~\cite{KreutzerPRS16,Pil2017number}, as well as a new algorithm
that is streamlined for bounded expansion classes and also achieves polynomial
bounds in $m$. We discuss these algorithms in more detail next.
\onlyfull{We will prove in Section~\ref{sec:lb} that the bounds of our new algorithm are close to optimal.}


\subsection{Distance trees}
First, we describe the algorithm that was introduced in~\cite{Pil2017number}.
We do so in
sufficient detail so that we can subsequently describe three of its variants which
we have implemented and included in our experimental evaluation.

Recall from Section~\ref{sec:prelims}
that a minor model of a graph $H$ in a graph~$G$ is a family $(I_u)_{u\in V(H)}$ of pairwise 
vertex-disjoint connected subgraphs of $G$
such that $u_1u_2 \in E(H)$ implies that there is $u_1'u_2' \in E(G)$ with $u_i' \in V(I_{u_i})$ for $i=1,2$.
A \emph{depth-$r$ minor} is a minor that admits a minor model where every set $I_u$ is of radius at most $r$.

\subsubsection{Description of the algorithm of Pilipczuk, Siebertz, and Toru\'{n}czyk}
On the theory side, the work of Pilipczuk, Siebertz, and Toru\'{n}czyk~\cite{Pil2017number} proved the following bounds.

\begin{theorem}[Theorem~1.5 of~\cite{Pil2017number}]
For all $r, t \in \mathbb{N}$ there is a polynomial $N$ with $N(m) = \Oh_{r,t}(m^{(4t+1)^{2rt}})$
such that the following holds.
Let $G$ be a graph without a $K_t$ as a depth-$\lfloor 9r/2 \rfloor$ shallow minor
and let $A \subseteq V(G)$ be a vertex subset of size at least $N(m)$ for a given $m$.
Then there exists a set $S \subseteq V(G)$ of size $|S| \leq t$ and a set $B \subseteq A \setminus S$
of size $|B| \geq m$ which is $r$-independent in $G \setminus S$. Moreover, given $G$ and $A$,
   such sets $S$ and $B$ can be computed in time $\Oh_{r,t}(|A| \cdot |E(G)|)$. 
\end{theorem}

For simplicity, we focus on the case $r=2$. First, observe that every graph from
a nowhere dense class contains large independent sets. By definition of a nowhere
dense class, some complete graph $K_t$ is excluded as a depth-$0$ minor, that
is, simply as a subgraph. Hence, Ramsey's Theorem immediately implies that 
if we consider any set $A\subseteq V(G)$ of size at least 
$\binom{t+m-2}{m-1}$, then there exists a set $B\subseteq A$ of size $m$ which
is independent (without deleting any elements). 
Furthermore, the proof of Ramsey's Theorem yielding this bound
is constructive and can easily be implemented. The difficult part is now to find
in a large independent set a large $2$-independent set, possibly after deleting 
a few elements (consider a family of stars to see that deletion may be necessary).

Assume now that $A$ is a large independent set. 
The idea is to arrange the elements of $A$ in a binary tree~$T$, which
we call a \emph{distance tree},
and prove that this tree contains a long path. From this path the set $B$ is extracted. 

\newcommand{\dau}{\mathrm{D}}
\newcommand{\son}{\mathrm{S}}
	
We identify the nodes of $T$ with words over the alphabet $\{0,1\}$, 
where $\epsilon$ corresponds to the root, and where for a word $w$ the
word $w0$ is its left and the word $w1$ is its right successor, respectively. Fix some
enumeration of the set $A$. We define $T$ by processing the elements 
of $A$ sequentially according to the enumeration. 
We start with the tree that has its root labeled with the first element of $A$. 
For each remaining element $a\in A$ we
execute the following procedure which results in adding a node with label $a$ 
to $T$. 
  
When processing the vertex $a$, do the following. Start with~$w$ being the empty word. While~$w$ is a node of $T$, repeat the following step: 
if the distance from $a$ to the vertex $b$ which is at the position corresponding
to $w$ in $T$ is at most~$2$, replace $w$ by $w0$, otherwise, 
replace~$w$ by $w1$. Once $w$ does not correspond to a node of $T$, 
extend $T$ by adding the node corresponding to $w$ and label it with~$a$. 
In this way, we have processed the element $a$, and now proceed to the next 
element of $A$ until all elements are processed. This completes the construction of $T$.
Thus, $T$ is a tree labeled with vertices of $A$, and every vertex of $A$ 
appears exactly once in~$T$.

Now, based on the fact that some complete graph $K_t$ is excluded as a depth-$2$
minor of $G$, it is shown that~$T$ contains a long path. This path either has
many left branches or many right branches. Take a subpath that has only left
branches or only right branches. Such a path corresponds to a set~$X$ such that
all elements have pairwise distance $2$, or all elements have pairwise distance
greater than $2$, that is, to a $2$-independent set. In the second case, we have
found the set $B$ that we are looking for. In the other case, we proceed to show
that there must exist an element $w\in V(G)$ that is adjacent to many elements
of $X$, i.e., $N(w)\cap X$ is large. We add the vertex $w$ to the set $S$ of elements
to delete and repeat the above tree-classification procedure with the set
$A'=N(w)\cap X$. It is shown that this process must stop after at most $t$ steps and
yields a set $B$ which is $2$-independent in $G-S$. 

The general case reduces to the case $r=1$ or $r=2$ if instead of starting with
an independent set~$A$ we start with an $i$-independent set $A_i$ and contract the
disjoint $i/2$ or $(i+1)/2$-neighborhoods of the elements of~$A_i$, respectively, to single vertices. Then 
one iteratively finds $i$-independent sets $A_1,A_2,\ldots, A_r$ for larger and larger radii. 

\subsubsection{Implementation details}

We have implemented three variants of the above method, which we
denote \tree, \treeshrink and \ldit. In all variants, we get a graph $G$, a vertex 
subset $A\subseteq V(G)$ and $r\in \N$ as input. We do not have the number
$m$ as input but we aim to find an $r$-independent
subset $B\subseteq A$ which is as large as possible while deleting as few
elements as possible. 

For the odd cases (which reduce to $r=1$
in the description above), in each variant we use a simple heuristic for finding
independent sets described in Section~\ref{sec:isnaive}. 

For more interesting even cases (which reduce to $r=2$ in the description 
above), \treeshrink computes a set of candidate solutions $(C,S)$ . Here, 
$C$ is a set which corresponds to a long path in the distance tree and 
$S$ is the set of vertices removed so far (for this set $C$). 
At every step we compute one candidate solution $(C, S)$,
remove a vertex $w$, i.e., move it to $S$, which has largest intersection
$|N(w)\cap A|$ and continue the process with $N(w)\cap A$ until $A$
becomes too small.
In the end, we output the best solution from the pool of 
collected solutions.

In the version denoted by \tree, we modify \treeshrink as follows. We let
$C$ be a~candidate for~a large $2$-independent set, which, however, 
we do not choose as a subset of the currently handled set $A$, but of the
original input set $A$. That is, we re-classify all distances of elements of 
the initial set $A$ in a distance tree 
with vertices $S$ that were deleted in later steps, to draw
the candidate $2$-independent set from a larger pool of vertices. 


Finally, in the \ldit version (least degree iterated) we do not find $2$-independent
sets based on the distance tree, but rather in a simple greedy manner as an independent
set in the graph $(G-S)^2[A]$. 


\subsection{Weak coloring numbers and uniform quasi-wideness}

A work of Kreutzer et al.~\cite{KreutzerPRS16} bound weak coloring numbers
with uniform quasi-wideness in graphs of bounded expansion. 
We include their approach in our comparison, as well as a new arguably simpler
algorithm inspired by their approach.

\subsubsection{Description of the algorithm by Kreutzer, Pilipczuk, Rabinovich, and Siebertz}
The following statement summarizes the theoretical bounds of the work of Kreutzer, Pilipczuk, Rabinovich, and Siebertz~\cite{KreutzerPRS16}.

\begin{theorem}[Theorem~4 of~\cite{KreutzerPRS16}]
Let $G$ be a graph and let $r, m \in \mathbb{N}$. Let $c \in \mathbb{N}$ be such that
$\wcol_r(G) \leq c$ and let $A \subseteq V(G)$ be a set of size at least $(c+1)\cdot 2^m$. 
Then there exists a set $S$ of size at most $c(c-1)$ and a set $B \subseteq A \setminus S$ 
of size at least $m$ which is $r$-independent in $G-S$.
\end{theorem}

Let $G$ be a graph, $A\subseteq V(G)$ and $m,r\in \N$ be given. 
First, fix some order $L\in \Pi(G)$ such that 
$|\WReach_r[G,L,v]|\leq c$ for every $v\in V(G)$ (for some constant $c$). 
Let $H$ be the graph with vertex set $V(G)$, where we put an edge $uv\in E(H)$ 
if and only if $u\in \WReach_r[G,L,v]$ or $v\in \WReach_r[G,L,u]$. 
Then $L$ certifies that $H$ is $c$-degenerate, and hence, assuming that
$|A|\geq (c+1)\cdot 2^m$, we can greedily
find an independent set $I\subseteq A$ of size $2^m$ in $H$. 
By the definition of the graph $H$, we have that $\WReach_r[G, L, v]\cap I=\{v\}$
for each $v\in I$. Now observe that for $v\in I$, deleting $\WReach_r[G, L, v] \setminus \{v\}$
from $G$ leaves $v$ at a~distance greater than $r$ (in $G - (\WReach_r[G, L, v] \setminus \{v\}))$ from all the other vertices of $I$.

Based on this observation, one follows the simple approach also used to prove
Ramsey's Theorem with exponential bounds. For each vertex $v$ of $I$ 
(in decreasing order, starting with the largest vertex with respect to $L$), 
we test whether $v$ is connected by a~path of length at most $r$ to more than 
half of the remaining vertices of $I$. If this is the case, we delete the set $\WReach_r[G,L,v]$ from $G$ (i.e., add it to $S$) and add the vertex $v$ to the 
set $B$. We continue with the subset of $I$ that had such a~connection to $v$
(which is, however, now separated by the deletion of $S$). Otherwise, $v$ is not
connected to more than half of the remaining vertices of $I$, in which case we simply
add $v$ to $B$ and do not delete anything. In this case, we continue the construction
with those vertices of $I$ that are not connected to $v$. It is proved that the first
case can happen at most $\wcol_r(G)\leq c$ many times, hence, in total we delete at most $c^2$ vertices and arrive at a~set $B$ with $m$ vertices that are pairwise at
distance greater than $r$ in $G-S$. 


We have implemented exactly the algorithm outlined above. 
We denote it by \mfcs. 

\subsubsection{A new algorithm}\label{ss:uqw-tgv}

\iflipics{\input{short-uqw-tgv-theory}}{Motivated by the rather conservative character of the algorithm of~\cite{KreutzerPRS16}
described above, we propose here a new algorithm (albeit inspired by~\cite{KreutzerPRS16}). Furthermore, in Section~\ref{sec:lb} we show an almost tight lower bound for the guarantees
of this algorithm in graphs excluding a fixed minor.

More formally, we show the following theorem.
\begin{theorem}\label{thm:uqw-tgv}
Assume we are given a graph $G$, a set $A \subseteq V(G)$,
integers $r \geq 1$ and $m \geq 2$, and an ordering~$\order$ of $V(G)$ with
$c = \max_{v \in V(G)} \left|\operatorname{WReach}_r[G, L, v]\right|$.
Furthermore, assume that
$|A| \geq 4\cdot (2cm)^{c+1}$. 
Then in polynomial time, one can compute sets $S \subseteq V(G)$
and $B \subseteq A \setminus S$ such that $|S| \leq c$, $|B| \geq m$, and $B$ 
is $r$-independent in $G-S$.
\end{theorem}

\begin{proof}
The algorithm iteratively constructs sets
$A = A_0 \supseteq A_1 \supseteq A_2 \supseteq \ldots$,
$\emptyset = S_0 \subseteq S_1 \subseteq S_2 \subseteq \ldots$, and
$\emptyset = B_0 \subseteq B_1 \subseteq B_2 \subseteq \ldots$,
maintaining the following invariants in every step $i$:
$B_i \subseteq A \setminus S_i$,
the set $B_i$ is an $r$-independent set in $G-S_i$, and every vertex of 
$A_i$ is within distance greater than $r$ from every vertex in $B_i$
in the graph $G-S_i$.

At step $i$, given $A_i$, $S_i$, and $B_i$, the algorithm proceeds as follows.
\begin{description}
\item[(stopping condition)] If $|A_i| \leq 2cm$, then stop and return $S = S_i$ and $B = B_i$.
\item[(growth step)] If $|A_i| > 2cm$ and
there exists $v \in A_i$ such that at most
$|A_i|/m$ vertices of $A_i$ are within distance at most $r$ from 
$v$ in $G-S_i$ (i.e., $|N_r^{G-S_i}[v] \cap A_i| \leq |A_i|/m$),
then move $v$ to $B_{i+1}$ and delete the conflicting vertices from $A_i$,
     that is set
\begin{align*}
A_{i+1} &= A_i \setminus N_r^{G-S_i}[v] \\
S_{i+1} &= S_i \\
B_{i+1} &= B_i \cup \{v\}.
\end{align*}
\item[(deletion step)] Otherwise, pick a vertex $z \in V(G) \setminus S_i$
that appears in a maximum number of weakly reachable sets of vertices of $A_i$.
That is, pick $z \in V(G) \setminus S_i$ maximizing the quantity
\[|\{v \in A_i ~|~ z \in \WReach_r[G, L, v]\}|.\] 
Insert $z$ into $S_{i+1}$ and restrict $A_i$ to vertices containing $z$
in their weak reachable sets. More formally,
\begin{align*}
A_{i+1} &= \{v \in A_i ~|~ z \in \operatorname{WReach}_r[G, L, v]\} \setminus \{z\} \\
S_{i+1} &= S_i \cup \{z\} \\
B_{i+1} &= B_i.
\end{align*}
\end{description}

Let us now analyze the algorithm. The fact that in the growth step we remove
from $A_{i+1}$ the vertices of~$A_i$ that are within distance at most $r$
from $v$ preserves the invariant that the distance between $A_i$ and~$B_i$
in $G-S_i$ is greater than $r$. This invariant, in turn, proves that 
$B_i$ is an $r$-independent set in $G-S_i$. It remains to show the bounds
on the sizes of $S$ and $B$. To this end, we show the following two claims.
\begin{claim}\label{cl:tgv-Z-bound}
At every step $i$, for every $z \in S_i$ and $v \in A_i$,
   we have that $z \in \WReach_r[G, L, v]$.
\end{claim}
\begin{proof}
The claim follows directly from the fact that in the deletion step,
    we restrict $A_{i+1}$ to be the set of those vertices
    of $A_i$ that have $z$ in their weak reachability set.
\cqed\end{proof}
\begin{claim}\label{cl:tgv-B-bound}
At every step $i$, if there is no vertex $v \in A_i$
with $|N_r^{G-S_i}[v] \cap A_i| \leq |A_i|/m$, 
then there exists $z \in V(G) \setminus S_i$ 
with at least $|A_i|/(cm)$ vertices $v \in A_i$\
satisfying $z \in \WReach_r[G, L, v]$.
\end{claim}
\begin{proof}
Let $v \in A_i$ be the least vertex of $A_i$ in the ordering $L$.
Since the growth step is not applicable, we have
that the set $X \coloneqq N_r^{G-S_i}[v] \cap A_i$ is of size
larger than $|A_i|/m$.
For every $x \in X$, fix a path $P_x$ of length at most~$r$ between
$x$ and $v$ in $G-S_i$, and let $z_x$ be the $L$-minimal vertex on this path.
The subpath of $P_x$ from $z_x$ to~$v$ shows that
$z_x \in \WReach_r[G, L, v]$
and the subpath of $P_x$ from $z_x$ to $x$ shows that
$z_x \in \WReach_r[G, L, x]$.
Since $|\WReach_r[G, L, v]| \leq c$, while $|X| > |A_i|/m$,
there exists $z \in \WReach_r[G, L, v]$ with 
$$|\{v \in A_i ~|~ z \in \WReach_r[G, L, v]\}| \geq 
\frac{|X|}{|\WReach_r[G, L, v]|} \geq \frac{|A_i|}{cm}.$$
This finishes the proof of the claim.
\cqed\end{proof}
Consequently, when the algorithm executes the deletion step,
  we have $|A_{i+1}| \geq |A_i|/(cm)-1$ (the $-1$ comes from the case $z \in A_i$).

In particular, we have that the last step of the algorithm is the growth step:
the deletion step executes only if $|A_i| > 2cm$, and then 
$|A_{i+1}| \geq |A_i|/(cm)-1 > 1$. 
Let $v$ be the vertex added to $B_{i+1}$ in this last growth step.
Then we have that
$S = S_{i+1} = S_{i} \subseteq \WReach_r[G, L, v]$.
Consequently, the algorithm executed at most $c$ deletion steps
and $|S| \leq c$.

For the bound on the size of set $B$, let $i$ be the index when the algorithm stopped, that is,
with $|A_i| \leq 2cm$. For every $0 \leq j < i$ that executed
a deletion step, we have
$$|A_{j+1}| \geq \frac{|A_j|}{cm} - 1 \geq \frac{|A_j|}{2cm}.$$
For every $0 \leq j < i$ that executed a growth step, we have
$$|A_{j+1}| \geq |A_j| - \frac{|A_j|}{m} = \left(1-\frac{1}{m}\right)|A_j|.$$
In particular, we have $A_i \neq \emptyset$ due to $m \geq 2$.
Consequently, since the algorithm executed $|S_i|$ deletion steps
and $|B_i|$ growth steps, we have
$$2cm \geq |A_i| \geq |A| \cdot \left(1-\frac{1}{m}\right)^{|B_i|} (2cm)^{-|S_i|}.$$
Hence, since $(1-1/m)^m \geq 1/4$ for every $m \geq 2$ and $|S_i| \leq c$,
if $|A| \geq 4 \cdot (2cm)^{c+1}$, then we have $|B_i| \geq m$.
This finishes the proof.
\end{proof}

\paragraph{Implementation details}
The actual implementation of the above algorithm differs in a number of aspects.
First, we found the threshold $|A_i|/m$ for the distinction between the growth step
and the deletion step too small in practice, despite working well in the proof
above. Moreover, experiments with this algorithm showed that it is
unstable in the sense that
small changes in this threshold can trigger big changes in the produced result
which are, a priori, hard to predict.
Because of that our implementation has a fixed constant $k$ and executes
the above algorithm with thresholds $\frac{1}{k+1}, \frac{2}{k+1}, \ldots, \frac{k}{k+1}$
and chooses the best result (we will address comparing different results later).

Second, the above algorithm can be modified so that the growth step is applied only in cases where the least vertex
of $A_i$ with respect to $L$ has only a small number of conflicts,
in which case we use that first vertex to enlarge $B$.
Note that such an algorithm also satisfies the theorem, because in the analysis of 
the algorithm
we used only the fact that if the growth step is not applicable, then this condition
is not satisfied for the first vertex of $A_i$. Such a variant is present in our implementation.

Third, in the proof above, the algorithm always applies the growth step when
the size of $A_i$ drops below the threshold $2cm$. This is a minor technical detail, and can be omitted at the cost of some more hassle in the proof (in the analysis
    of the last steps of the algorithm) and somewhat worse bounds for $|S|$ and $|A|$. In the implementation, we do not have this threshold, but instead
we roll back the unnecessary deletion steps that were performed by the algorithm
near the end of the execution. It is straightforward (but a bit more tedious)
  to adapt the above analysis to this variant.


\paragraph{Implemented variants}
We have implemented three variants of the above described method, which
we denote \tgva, \tgvb and \tgvc. In the outlined algorithm, when we consider 
a vertex $v$, we compute the set of vertices from $A$ conflicting with~$v$.
In \tgva, we consider two vertices to be conflicting if their 
$\WReach_r$ sets intersect. In \tgvb and \tgvc, two vertices are considered 
to be conflicting if the distance between them in the remaining part of the graph 
is at most $r$. Moreover, \tgvc after every step tries to fill its partial solution 
with the heuristic described in Section~\ref{sec:isnaive} to find an independent set in 
$(G-S)^r \cap A$, where $S$ is a set of already removed vertices. 

}

\subsection{Other naive approaches and heuristic optimizations}\label{sec:isnaive}
Since uniform quasi-wideness for $r=1$ is exactly finding independent sets, it makes sense to
include heuristics for finding independent sets as a~baseline. Moreover,
the problem of finding independent sets is also used as a~subroutine in the 
approach based on distance trees.
We used the following simple greedy algorithm to find independent sets. 
As long as our graph
is nonempty, take any vertex that has the smallest degree, add it
to the independent set and remove it and its neighbors from the graph.

The following algorithm is what we came up with as a~naive but 
reasonable heuristic for larger values of $r$.
For every number $k \in \{0, 1, \ldots, K\}$ (where $K$ is some hardcoded constant)
computes the biggest independent set in the graph
$(G - S_k)^r[A]$ using
the greedy procedure described above, where~$S_k$ is a set of $k$ vertices with biggest degrees.
This heuristic 
is based on the fact that independent sets
in $G^r$ correspond to $r$-independent sets in $G$. Without any other 
knowledge about the graph, vertices with the biggest degree seem to be the 
best candidates to be removed. In the end,
we output the best solution obtained in this manner. 
In the following, we abbreviate this approach as \ldpow
(least degree on power graph).

We remark that the used least degree heuristic is probably the simplest one for finding a maximum independent set
in a graph, but there are multiple better solutions available, both heuristic~\cite{mis-heur} and exact~\cite{cliquer,Lamm0SWZ19}.
Exploring the usage of more sophisticated algorithms in place of the least degree heuristic is beyond the scope of this work,
and, judging from the good performance of the heuristic described in this section, an interesting direction for future work.

\subsection{Comparing different results}

Uniform quasi-wideness is a~two-dimensional measure: we have to measure both
the size $m$ of the $r$-independent set $B$ which we desire to find, 
as well as the size $s(r)$ of vertices to be deleted. In order to compare the 
performance of our studied methods we propose the following approach that
arises from applications of uniform quasi-wideness in several algorithms~\cite{DawarK09,DrangeDFKLPPRVS16,Pil2017number,Siebertz17}.

Let $G, A \subseteq V(G), r \in \N$ be an input to any of our algorithms
(note that none of our algorithms takes the target size of the $r$-independent 
set as input) and let $S \subseteq V(G)$ and $B \subseteq A \setminus S$
such that $B$ is $r$-independent in $G - S$ be its output. 
Let us define $\pi_r[v, S]$ -- the \textit{$r$-distance profile of $v$ on $S$} --
as the function from $S$ to $\{0, 1, \ldots, r, \infty\}$ so that $\pi_r[v, S](a) = dist_G(v, a)$ if this distance is at most $r$, and $\pi_r[v, S](a) = \infty$ otherwise. 

The performance of the algorithms~\cite{DawarK09,DrangeDFKLPPRVS16,Pil2017number,Siebertz17} 
strongly depends on the size of the largest
equivalence class on $B$ defined by $u\sim v$ if $\pi_r[u,S]=\pi_r[v,S]$ for $u,v\in B$. 
Indeed, a recurring theme in these algorithms is to argue that if an equivalence class is sufficiently large, then an arbitrary vertex of the class is \emph{irrelevant} for the problem;
for example, the main argument of the kernelization algorithm for \textsc{Dominating Set}~\cite{DrangeDFKLPPRVS16} asserts that, given large equivalence class $B$, for every $v \in B$
one can lift the requirement to dominate $v$ without changing the answer to the problem. 

We hence decided to use the size of the largest equivalence class in the above
relation as the scoring function to measure the 
performance of our algorithms. 
Note that the number of different $r$-distance profiles is bounded by $(r + 2)^{|S|}$, so 
if $r$ is fixed and $|S|$
is bounded then the number of different $r$-distance profiles is also bounded, 
so having a big $r$-independent set
implies having a big subset of this set with equal $r$-distance profiles on $S$.

This well defined scoring function makes it possible to compare the results of the 
algorithms. Furthermore, in our code the implementation of the scoring function 
can be easily exchanged, so if different scoring functions are preferred, re-computation
and re-evaluation
is easily possible.


\section{Experimental setup}\label{sec:setup}

\subsection{Hard- and Software}

The experiments on generalized coloring numbers have been performed on
an Asus K53SC laptop
with Intel\textregistered{} Core\texttrademark{} i3-2330M CPU @ 2.20GHz x 2 processor and with 7.7 GiB of
RAM. Weak coloring numbers of a larger number of graphs for the statistics in
Section~\ref{sec:statistics_wcol} (presented without running times)  were
produced on a cluster at the Logic and Semantics Research Group, Technische
Universit\"{a}t Berlin.
The experiments on uniform quasi-wideness have been performed on a cluster of
16 computers at the Institute of Informatics, University of Warsaw.
Each machine was equipped with Intel Xeon E3-1240v6 3.70 GHz processor and 16 GB RAM.
All machines shared the same NFS drive. Since the size of the inputs and
outputs to the programs is relatively small, the network communication was negligible for tests with substantial running times.
The dtf implementation has been done in Python, while all other code in C++ or C.
The code is available at~\cite{our-repo,fnp-webpage}.

\subsection{Test data}

Our dataset consists of a number of graphs from different sources.
\vspace{-2mm}

\begin{description}[leftmargin=*]
\item[Real-world data] We collected appropriately-sized networks 
  from several collections~\cite{GephiData,GraphBaseData,SnapData,PajekData,NetworkRepo,konect}.
  Our selection contains classic social networks \cite{Karate,BrightKite},
  collaboration networks \cite{Densification,Netscience,CondMat}
  contact networks \cite{SpDataSchool,Dolphins},
  communication patterns~\cite{Densification,Gnutella,EnronData,Slashdot,Epinions,Polblogs},
  protein-protein interaction~\cite{Yeast},
  gene expression~\cite{Diseasome},
  infrastructure~\cite{Power}, tournament data \cite{Football}, and neural networks~\cite{CElegans}.
  We kept the names assigned to these files by the respective source.

\item[PACE 2016 Feedback Vertex Set] 
The Parameterized Algorithms and Computational Experiments Challenge
is an annual programming challenge started in 2016 that aims
to \emph{investigate the applicability of algorithmic ideas studied and developed in the subfields of multivariate, fine-grained, parameterized, or fixed-parameter tractable algorithms} (from the PACE webpage). 
In the first edition, one of the tracks focused on the \textsc{Feedback Vertex Set} problem~\cite{PACE2016},
   providing 230 instances from various sources and of different sizes.
We have chosen a number of instances with small feedback vertex set number,
   guaranteeing their very strong sparsity properties (in particular, low
       treewidth).
In our result tables, they are named \texttt{fvs???}, where \texttt{???}
is the number in the PACE 2016 dataset.
\item[Random planar graphs]
In their seminal paper, Alber, Fellows, and Niedermeier~\cite{AlberFN04}
initiated the very fruitful direction of developing \enlargethispage{\baselineskip}
of polynomial kernels (preprocessing routines
    rigorously analyzed through the framework of parameterized complexity)
in sparse graph classes
by providing a linear kernel for \textsc{Dominating Set} in planar graphs.
\iflipics{}{\textsc{Dominating Set} soon turned out to be the pacemaker of the development
of fixed-parameter and kernelization algorithms in bounded expansion
and nowhere dense graph
classes~\cite{AmiriMRS17,DawarK09,DrangeDFKLPPRVS16,Dvorak13}.}
In~\cite{AlberFN04}, an experimental evaluation is
conducted on random planar graphs
generated by the LEDA library~\cite{LEDA}. 
We followed their setup and included a number of random planar graphs
with various size and average degree.
In our result tables, they are named \texttt{planarN}, where \texttt{N}
stands for the number of vertices.
\item[Random graphs with bounded expansion]
A number of random graph models has been shown to produce almost 
surely graphs of bounded expansion~\cite{ComplexNetworks}.
We include a number of graphs generated
by O'Brien and Sullivan~\cite{OBrienSullivan} using the following
models: the stochastic block model (\texttt{sb-?} in our dataset)~\cite{SB}
and the Chung-Lu model with households (\texttt{clh-?}) and without
households (\texttt{cl-?})~\cite{ChungL03}.
We refer to~\cite{ComplexNetworks,OBrienSullivan} for more discussion
on these sources.
\end{description}
The graphs have been partitioned into four groups, depending on their size:
the \texttt{small} group gathers graphs up to $1\,000$ edges,
    \texttt{medium} between $1\,000$ and $10\,000$ edges,
    \texttt{big} between $10\,000$ and $48\,000$ edges,
    and \texttt{huge} above $48\,000$ edges.
The random planar graphs in every test group have respectively
$900$, $3\,900$, $21\,000$, and $150\,000$ edges.
The whole dataset is available for download at~\cite{fnp-webpage}.
Table~\ref{tb:teststats} gathers basic statistics about test groups.
For every test group, the repository~\cite{our-repo} offers a CSV file \texttt{group\_test\_stats.csv}
with a detailed breakdown.

\begin{table}[htb]
\begin{center}
  \begin{tabular}{@{}l|cccc|cccc@{}}
\multirow{2}{*}{group} & \multicolumn{4}{c|}{$|V(G)|$} & \multicolumn{4}{c}{$|E(G)|$} \\
& min & med & avg & max 
& min & med & avg & max \\\hline
small & 34 & 115 & 222.52 & 620 & 62 & 612 & 520.61 & 930 \\
medium & 235 & 1302 & 1448.44 & 4941 & 1017 & 3032 & 3343.44 & 8581 \\
big & 1224 & 7610 & 7963.64 & 16264 & 10445 & 21000 & 19519.00 & 47594 \\
huge & 3656 & 27775 & 34598.69 & 77360 & 48130 & 186940 & 237300.06 & 546487 \\
\end{tabular}
\caption{Basic statistics of test groups. avg stands for average, med stands for median.}\label{tb:teststats}
\end{center}
\end{table}


\section{Weak coloring numbers: results}\label{sec:wcol-results}

\subsection{Quality ratio}
As already discussed in the introduction, for all graphs in our data set we
do not know the exact (optimal) value of the weak coloring number
and we do not know how to compute them efficiently even in the data set consisting
of small graphs.

Thus, to evaluate the quality of each algorithm, we proceed as follows. 
For each graph in the data set, we take all the orderings produced by all algorithms
in the experiment (including the improved orderings produced by the local search routine)
and take note of the smallest weak coloring number encountered. 
This number is the best known upper bound on the weak coloring number of the graph
in question and we grade each algorithm by the ratio of the weak coloring number
of the ordering produced by the algorithm to this best known upper bound. 
That is, in this section, the term \emph{ratio} always refers to the ratio
to the best known upper bound on the weak coloring number of the graph in question.

In Table~\ref{tb:wcolstats} we gather basic statistics on the values of 
the weak coloring number in different data sets. 
Note that in our repository~\cite{our-repo} one can find CSV files with the values
of the weak coloring number of each ordering produced by each algorithm on each test.

\begin{table}[htb]
\begin{center}
  \begin{tabular}{@{}lc|cccc@{}}
\multirow{2}{*}{group} &
\multirow{2}{*}{radius} &
  \multicolumn{4}{c}{$\wcol_r$} \\
& & min & med & avg & max \\\hline
\multirow{5}{*}{small}
 & 1 & 3 & 5 & 5.70 & 10 \\
 & 2 & 5 & 10 & 12.87 & 38 \\
 & 3 & 6 & 15 & 19.00 & 65 \\
 & 4 & 8 & 18 & 22.39 & 74 \\
 & 5 & 8 & 21 & 24.39 & 74 \\
\hline
\multirow{5}{*}{medium}
 & 1 & 3 & 6 & 8.03 & 34 \\
 & 2 & 6 & 19 & 23.50 & 118 \\
 & 3 & 9 & 33 & 45.50 & 143 \\
 & 4 & 12 & 40 & 67.62 & 172 \\
 & 5 & 15 & 40 & 84.19 & 243 \\
\hline
\multirow{5}{*}{big}
 & 1 & 3 & 7 & 10.41 & 37 \\
 & 2 & 6 & 31 & 32.59 & 162 \\
 & 3 & 11 & 59 & 84.86 & 301 \\
 & 4 & 14 & 81 & 182.18 & 1021 \\
 & 5 & 19 & 116 & 285.27 & 1734 \\
\hline
\multirow{5}{*}{huge}
 & 1 & 4 & 38 & 52.50 & 239 \\
 & 2 & 10 & 223 & 279.75 & 885 \\
 & 3 & 18 & 531 & 730.81 & 2595 \\
 & 4 & 23 & 884 & 1349.88 & 4564 \\
 & 5 & 27 & 1138 & 2109.88 & 7706 \\
\hline
\end{tabular}
\caption{Basic statistics of best known weak coloring numbers. avg stands for average, med stands for median.}\label{tb:wcolstats}
\end{center}
\end{table}

\subsection{Fine-tuning flat decompositions}

\begin{table}[bht]
\centering%
\begin{tabular}{l>{\hspace*{9pt}}ll>{\hspace*{9pt}}ll>{\hspace*{9pt}}l}
option & \hspace*{-9pt}\shortstack{average \\ ratio}
& option & \hspace*{-9pt}\shortstack{average \\ ratio}
& option & \hspace*{-9pt}\shortstack{average \\ ratio} \\ \toprule
BFS/\eqref{flat:1} & 1.159 & DFS/\eqref{flat:1} & 1.156 & SORT/\eqref{flat:1} & 1.072 \\
BFS/\eqref{flat:2} & 1.131 & DFS/\eqref{flat:2} & 1.117 & \textbf{SORT/\eqref{flat:2}} & \textbf{1.039} \\
BFS/\eqref{flat:3} & 1.147 & DFS/\eqref{flat:3} & 1.135 & SORT/\eqref{flat:3} & 1.054 \\ \midrule
$\overline{\textrm{BFS}}$/\eqref{flat:1} & 1.363 & $\overline{\textrm{DFS}}$/\eqref{flat:1} & 1.368 & $\overline{\textrm{SORT}}$/\eqref{flat:1} & 1.41 \\
$\overline{\textrm{BFS}}$/\eqref{flat:2} & 1.277 & $\overline{\textrm{DFS}}$/\eqref{flat:2} & 1.291 & $\overline{\textrm{SORT}}$/\eqref{flat:2} & 1.329 \\
$\overline{\textrm{BFS}}$/\eqref{flat:3} & 1.309 & $\overline{\textrm{DFS}}$/\eqref{flat:3} & 1.324 & $\overline{\textrm{SORT}}$/\eqref{flat:3} & 1.36 \\\bottomrule \vspace{4pt}
\end{tabular}
\caption{Comparison of different flat decomposition variants: sorting
vertices of the new blobs $B_i$ by the BFS, DFS, by degree (non-increasing),
 or these orders reversed; the second coordinate refers to
 the choice of the root vertex: \eqref{flat:1} maximizing the number of neighbors
 already processed, \eqref{flat:2} maximizing degree in~$U$,
 \eqref{flat:3} as previous, but only among neighbors of already processed vertices.
The value is the average of the ratios to the best generalized
coloring numbers found by all versions of this algorithm.}\label{tb:flat}\vspace{-.5cm}
\end{table}

\noindent
As discussed in Section~\ref{ss:flat}, we have experimented with a number of
variants of the flat decompositions approach, with regards to the choice of
the next root vertex and the internal order of the vertices of the next $B_i$.
The results for the \texttt{big} dataset are presented in Table~\ref{tb:flat}.
They clearly indicate that (a) all reversed orders performed much worse, and
(b) among other options, the best is to sort the vertices of
a new $B_i$ non-increasingly by degree and choose as the next root the vertex
of maximum degree.
In the subsequent tests, we use this best configuration for comparison with
other approaches.

\subsection{Comparison of all approaches}


Table~\ref{tb:wcol} presents the results of our experiments on all test
instances and all approaches, summarized as follows:
\begin{description}
\item[dtf] dtf-augmentations with the respective radius~$r$ supplied as the distance bound;
\item[flat] the best configuration of the flat decompositions approach (see previous section);
\item[treedepth] the treedepth approximation heuristic;
\item[treewidth] the treewidth heuristic;
\item[degree sort] the heuristic which sorts the vertices non-increasingly by degree.
\item[WReach] greedy approach constructing the ordering from left to right, picking at every step a vertex with the largest potential weakly reachable set;
\item[SReach] greedy approach constructing the ordering from right to left, picking at every step a vertex with the smallest potential strongly reachable set.
\end{description}
Out of all simple heuristics (c.f.\ Section~\ref{ss:simple}) degree
sorting was supreme and we skip the results of inferior
heuristics (see~\cite{our-repo,fnp-webpage} for full data).
Interestingly, this heuristic also outperformed most other (much
more involved) approaches.
In all cases, the greedy approaches described in sections~\ref{ss:WReachLeft} and~\ref{ss:SReachRight}
outperformed the rest, with the left-to-right greedy algorithm based on weakly reachable sets being the best for smaller radii
and the right-to-left greedy algorithm based on strongly reachable sets being the best for larger radii.

Interestingly, on small graphs, the treewidth heuristic returns competitive results. An explanation why the
treewidth heuristic is better on smaller graphs $G$ might be that $\tw(G) =
\col_\infty(G)$ and on small graphs the difference between $\col_\infty(G)$ and
$\col_r(G)$ for the considered $r$ is not that big. However, this does not
explain why treedepth does not perform better than treewidth. (Recall that $\td(G)
= \wcol_\infty(G)$.)
It is worth observing that on larger graphs (the \texttt{big} group)
the performance of the flat decomposition matches or outperforms the one
of the treewidth heuristic for radii $r=2,3,4$.
However, the treewidth heuristic outperforms all
approaches with proved guarantees for $r=5$ on test sets up to the
\texttt{big} group.

Table~\ref{tb:wcol} gathers total running time of our programs on
discussed data sets. These results clearly indicate large discrepancy
between consumed resources for different approaches.
Out of the approaches with provable guarantees on the output
coloring number, the flat decompositions approach is clearly the most efficient.

\afterpage{%
\clearpage%
\begin{landscape}
\begin{table}
\centering
\setlength{\aboverulesep}{0pt} 
\setlength{\belowrulesep}{0pt} 
\def\spacetop{\rule{0pt}{12pt}} 
\begin{tabular}{@{}lcAcAcAcAcAcAcAc@{}}
\chead{tests} & \chead{$r$} & \mhead{dtf} & \mhead{flat} & \mhead{treedepth} & \mhead{treewidth} & \mhead{degree sort} & \mhead{WReach} & \mhead{SReach} \\ \toprule
\mrow{small}
& \spacetop 2 & 1.275 & 0:04.74 & 1.289 & \mrow{0:00.02} & 1.514 & \mrow{0:09.57} & 1.202 & \mrow{0:00.35} & 1.267 & \mrow{0:00.32} & 1.083 & 0:00.04     & 1.155 & 0:00.06 \\
& 3           & 1.513 & 0:04.18 & 1.307 &                & 1.516 &                & 1.186 &                & 1.276 &                & 1.100 & 0:00.05     & 1.107 & 0:00.02 \\
& 4           & 1.627 & 0:04.70 & 1.346 &                & 1.447 &                & 1.184 &                & 1.269 &                & 1.177 & 0:00.07     & 1.075 & 0:00.03 \\
& 5           & 1.749 & 0:05.61 & 1.382 &                & 1.440 &                & 1.187 &                & 1.290 &                & 1.226 & 0:00.06     & 1.084 & 0:00.07 \\[2pt]
\midrule
\mrow{medium}
& \spacetop 2 & 1.326 & 0:20.41 & 1.541 & \mrow{0:01.13} & 2.474 & \mrowNA        & 1.751 & \mrow{0:18.36} & 1.285 & \mrow{0:00.85} & 1.085 & 0:00.60     & 1.191 &   0:00.98    \\
& 3           & 1.440 & 0:44.33 & 1.655 &                & 2.240 &                & 1.513 &                & 1.271 &                & 1.116 & 0:00.71     & 1.104 &   0:00.96    \\
& 4           & 1.698 & 1:11.08 & 1.672 &                & 1.974 &                & 1.343 &                & 1.285 &                & 1.089 & 0:01.11     & 1.058 &   0:01.14    \\
& 5           & 1.777 & 1:37.55 & 1.660 &                & 1.816 &                & 1.232 &                & 1.294 &                & 1.163 & 0:01.52     & 1.040 &   0:01.34    \\[2pt]
\midrule
\mrow{big}
& \spacetop 2 & 1.304 & \NA     & 1.706 & \mrow{0:17.54} & \NAgr & \mrowNA & 2.773 & \mrowNA & 1.400 & \mrow{0:02.28} & 1.075 & 0:03.73 & 1.202 & 0:11.32 \\
& 3           & 1.528 & \NA     & 1.796 &                & \NAgr &         & 2.452 &         & 1.356 &                & 1.084 & 0:06.61 & 1.185 & 0:12.40 \\
& 4           & \NAgr & \NA     & 1.827 &                & \NAgr &         & 1.862 &         & 1.382 &                & 1.097 & 0:14.57 & 1.117 & 0:16.02 \\
& 5           & \NAgr & \NA     & 1.777 &                & \NAgr &         & 1.495 &         & 1.329 &                & 1.345 & 0:25.35 & 1.042 & 0:24.80 \\[2pt]
\midrule
\mrow{huge}
& \spacetop 2 & \NAgr & \NA     & 2.124 & \mrow{4:14.11} & \NAgr & \mrowNA & \NAgr   & \mrowNA & 1.432 & \mrow{0:16.91}  & 1.086 &  1:07.98    & \NAgr & \NA  \\
& 3           & \NAgr & \NA     & 2.618 &                & \NAgr &         & \NAgr   &         & 1.342 &                 & 1.152 & \NA         & \NAgr & \NA       \\
& 4           & \NAgr & \NA     & 2.506 &                & \NAgr &         & \NAgr   &         & 1.293 &                 & \NAgr & \NA         & \NAgr & \NA       \\
& 5           & \NAgr & \NA     & 2.389 &                & \NAgr &         & \NAgr   &         & 1.234 &                 & \NAgr & \NA         & \NAgr & \NA       \\[2pt]
\bottomrule
\end{tabular}%
\caption{\emph{Gray columns}: Comparison of the main approaches and their average
  ratio to the best found coloring number.
  Some of the approaches did not finish in time on larger graphs or ran out of memory.
  \emph{White columns}:
  Total running time of the main approaches.
  Note that for some approaches the ordering (and thus running time)
  is independent of the radius.}\label{tb:wcol}
\end{table}%
\begin{table}
\centering
\setlength{\aboverulesep}{0pt} 
\setlength{\belowrulesep}{0pt} 
\def\spacetop{\rule{0pt}{12pt}} 
\begin{tabular}{@{}ccAcAcAcAcAcAcAc@{}}
\chead{tests} & \chead{radius} & \mhead{dtf} & \mhead{flat} & \mhead{treedepth} & \mhead{treewidth} & \mhead{degree sort} & \mhead{WReach} & \mhead{SReach} \\\toprule
\mrow{small}
& 2 & 1.155 & \mrow{16.7\%} & 1.060 & \mrow{16.9\%} & 1.172 & \mrow{15.2\%} & 1.087 & \mrow{7.0\%} & 1.053  & \mrow{16.2\%} & 1.069 & \mrow{6.7\%} & 1.063 & \mrow{7.3\%} \\
& 3 & 1.256 &               & 1.100 &               & 1.263 &              & 1.122 &              & 1.065  &                & 1.053 & & 1.041 &\\
& 4 & 1.343 &               & 1.105 &               & 1.299 &              & 1.145 &              & 1.066  &                & 1.096 & & 1.032 &\\
& 5 & 1.480 &               & 1.148 &               & 1.325 &              & 1.165 &              & 1.100  &                & 1.136 & & 1.056 &\\
\midrule
\mrow{medium}
& 2 & 1.207 & \mrow{13.9\%} & 1.151 & \mrow{21.4\%} & 1.224 & \mrow{30.9\%} & 1.149 & \mrow{15.3\%} & 1.024 & \mrow{17.1\%} & 1.070 & \mrow{2.8\%} & 1.012 & \mrow{9.9\%} \\
& 3 & 1.249 &               & 1.159 &               & 1.354 &               & 1.167 &               & 1.062 &               & 1.110 & & 1.011 & \\
& 4 & 1.530 &               & 1.359 &               & 1.440 &               & 1.216 &               & 1.087 &               & 1.108 & & 1.006 & \\
& 5 & 1.582 &               & 1.424 &               & 1.505 &               & 1.226 &               & 1.118 &               & 1.161 & & 1.021 & \\
\midrule
\mrow{big}
& 2 & 1.172 & \mrowNA & 1.196 & \mrow{24.4\%} & \NAgr & \mrowNA & 1.268 & \mrow{24.3\%} & 1.091 &  \mrow{18.5\%} & 1.087 & \mrow{1.6\%} & 1.023 & \mrow{11.5\%} \\
& 3 & 1.321 && 1.239  && \NAgr && 1.415  && 1.097 & & 1.105 & & 1.019 & \\
& 4 & \NAgr   && 1.390 && \NAgr && 1.434 && 1.145 & & 1.123 & & 1.020 & \\
& 5 & \NAgr   && 1.438 && \NAgr && 1.387 && 1.164 & & 1.177 & & 1.010 & \\
\bottomrule \vspace*{1pt}
\end{tabular}
\caption{\emph{Gray columns}: Comparison of average ratio
  after local search.
  \emph{White columns}: Relative improvement of local search for ordering output
  by the studied approaches.
  }\label{tb:ls-merged}\vspace*{-.5cm}
\end{table}%
\end{landscape}%
\clearpage%
}

Note that we applied different timeout policies for generating different data.
For generating time of execution and for applying local search
we set the timeout to be 1 minute, however
for generating orders and wcol numbers we set the timeout to be 5 minutes,
but for the sake of completeness we sometimes allowed some
programs to run longer.

In summary, on our data sets the greedy approaches of Sections~\ref{ss:WReachLeft} and~\ref{ss:SReachRight} produce the best results and have competitive running times.
If one looks for something faster, the simple sort-by-degrees heuristic is consistently the fastest and produces good results.
It is worth noting that on the smallest graphs it is outperformed by the treewidth heuristic.

We remark here that it is simple to ``fool'' the degree-sorting heuristic
by adding multiple pendant vertices of degree one and thus forcing it to take
an arbitrarily bad ordering, but such adversarial obstacles seem to be
absent in real-world graphs.
If one is to choose an algorithm with provable guarantees, the discussed
variant of the flat decompositions approach appears to be the best choice.

\subsection{Local search}

\noindent
In a second round of experiments we applied a simple local-search routine
that, given an ordering output by one of the approaches, tries to improve it
by moving vertices with the largest weakly reachable sets earlier in the
ordering. The white columns in Table~\ref{tb:ls-merged} show how local search
improved orderings output by discussed approaches, and the gray columns show
average ratios of orderings improved by local search. Two
remarks are in place.

First, regardless of how the ordering was computed, a
local search step almost always significantly improves the ordering.
The main exception is the case of the left-to-right greedy approach of Section~\ref{ss:WReachLeft}, which can be explained by the fact
that already the greedy algorithm explicitly optimizes sizes of the same sets as the local search heuristic.
We have no good explanation on why local search is significantly less effective on the orderings output by the treewidth heuristic for bigger radii.

Second, in general the local search step
does \emph{not} improve the orderings enough to change the relative order of the
performance of the base approaches.
However, there are few exceptions.
The poor performance of local search on the output of the left-to-right greedy algorithm of Section~\ref{ss:WReachLeft}
puts it behind the right-to-left greedy algorithm of Section~\ref{ss:SReachRight} and the sort-by-degrees heuristic.
Moreover, on the \texttt{medium} group
the treewidth heuristic gave better results than the sort-by-degrees heuristic on $r=5$, however degree sort regained the lead
after application of local search due to its low performance on larger radii
for treewidth heuristic.

We therefore recommend the local search
improvement as a relatively cheap post-processing improvement to any existing algorithm.
The combination of the right-to-left greedy algorithm based on strongly reachable sets (described in Section~\ref{ss:SReachRight})
with the local search improvement is the clear winner in our final comparison.
If one needs something faster, we recommend the simple degree sort heuristic.

\subsection{Correlation of weak coloring numbers with other parameters}
\label{sec:statistics_wcol}
\noindent
While it is undeniable that weak coloring numbers have immense algorithmic
power from a theoretical perspective, the efficient computation of such weak
coloring orders is only one component to leverage them in practice: we also
need these numbers to be reasonably low. So far, this had only been
established on a smaller scale~\cite{ComplexNetworks,FelixThesis} for a
related measure. Here, we computed the weak coloring number for $r \in
\{1, \ldots,5\}$ for 1675 real-world networks from various
sources~\cite{konect,SnapData,NetworkRepo,PajekData,GephiData}.
Figure~\ref{fig:corr} summarizes our findings for~$r \in \{1,3,5\}$: we
find a modest correlation with~$n$ and a significant correlation with~$m$.
The correlation with $n$ becomes quite pronounced for~$r = 5$; the probable reason being
that for all networks involved $\log n \leq 10$. Still, even in the worst
examples $\wcol_5$ is at least one order of magnitude smaller than~$n$ or~$m$.
We further see a high correlation between $\wcol_1$ and the average
degree~$\bar d$ which vanishes for larger radii. It is no big surprise
that~$\bar d$ and the degeneracy~$\wcol_1$ are highly correlated since these
values are only far apart in graphs with highly inhomogeneous densities. The decrease for larger radii indicates that vertices of high degree do not tend to build large highly connected clusters.

The low dependence on the maximum degree 
confirms the
findings of \cite{ComplexNetworks}: the exact shape of the degree
distribution's tail is much more relevant than the singular value of the
maximum degree. Finally, note that in our graphs the degeneracy
$\wcol_1$ practically does not grow with~$n$.
\begin{figure}[htb]
  \centering
  \includegraphics[scale=.25]{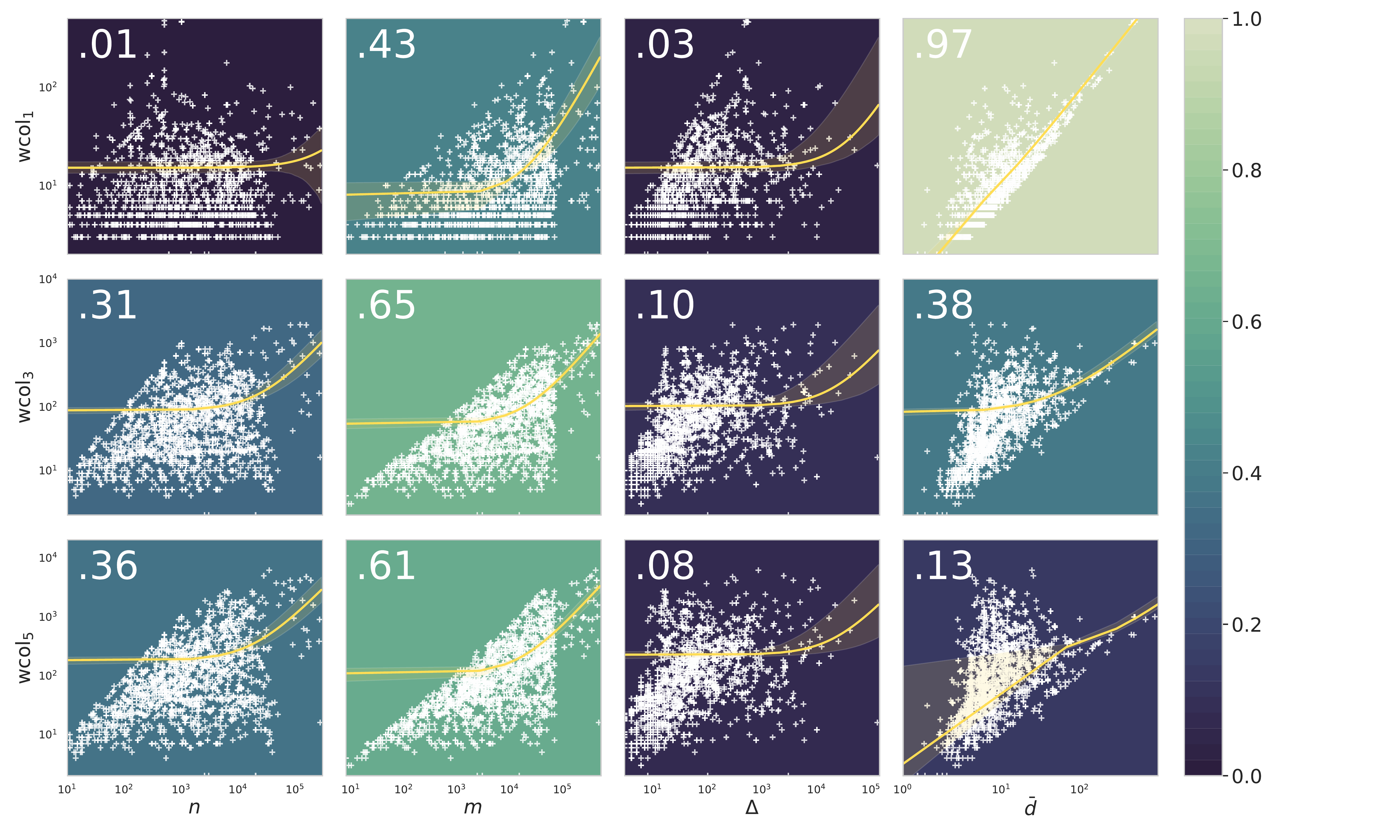}
  \caption{\label{fig:corr}%
    Correlation of $\wcol$ (computed using the SReach heuristic where possible,
    otherwise resorting to the degree sort heuristic) with graph size,
    maximum degree and average degree of 1703 real-world graphs.
    The background shade and number reflect the correlation of the two respective measures,
    superimposed is a log-log plot of the measurements. The yellow lines are linear regressions with
    lightly shaded confidence intervals.
  }
\end{figure}


\section{Uniform quasi wideness: results}\label{sec:uqw-results}

Table~\ref{tb:uqw} gathers aggregated data from our experiments on the \texttt{medium} dataset.
(Full data can be downloaded from~\cite{our-repo,fnp-webpage}.)
Every tested algorithm has been run on every test with timeout $10$ minutes
and with radii $r \in \{2,3,4,5\}$ and with the starting set either $A = V(G)$ or a random subset
of $20\%$ of vertices of $V(G)$.

Data indicate the simple heuristic, \ldpow, as the best choice in most scenarios,
as it has always best or nearly-best total score and runs relatively quickly.
The third variant of the new algorithm \tgvc has comparable results, but is inefficient and does not
finish within the timeout. 
Other variants \tgva and \tgvb as well as \mfcs are significantly outperformed by other approaches.
Out of other approaches with provable guarantees, the variants \tree, \treeshrink, and \ldit provide results in most cases less
than $10\%$ worse than the heuristic \ldpow, with \treeshrink being consistently worse.

Our initial inspiration for designing the new algorithm (variants \tgva, \tgvb, and \tgvc) was to avoid conservative deletion
steps in the algorithm \mfcs. On one hand, this particular goal has been achieved, as the deletion sets output 
by the algorithms \tgva, \tgvb, and \tgvc are of order of magnitude smaller than the ones output by the algorithm \mfcs.
However, the overall quality of \tgva and \tgvb turned out to be still poor compared to the variants based on distance trees,
and \tgvc is clearly the slowest of the algorithms while producing results comparable with the best other algorithms when
it finished within reasonable time. 
This suggests the following explanation.
The main combinatorial idea of the algorithms \mfcs, \tgva, \tgvb, and \tgvc is, upon a deletion step, to restrict
to the weakly reachable set of the deleted vertex. This allows to provide a strong theoretical guarantee on the size
of the deletion set but in practice turns out to be too conservative, as witnessed by the results of the algorithm \tgvc.
This algorithm, by additionally performing a heuristic step of finding an independent set after every deletion step, 
escapes this pitfall, but at too large running time cost.

We also remark that the total independent set size found by \mfcs is large, by far the largest among all algorithms for larger radii. 
However, this comes up at the cost of a very large deletion set which, in turn, makes the final score low. 
Note that it is very simple to come up with a large $r$-independent set if one does not care about the size of the deletion set:
just find (e.g., by a greedy heuristic) a large $1$-independent set $B$ and delete $V(G) \setminus B$, making $B$
$r$-independent for every $r$. 
As almost all known algorithmic usages of uniform quasi-wideness focus on the largest equivalence class of the distance profile,
we think this should be the main factor in evaluating uniform quasi-wideness algorithms and, consequently, we evaluate the performance
of the algorithm \mfcs as rather poor.

To sum up, our experiments show that the simple heuristic \ldpow gives best results, but if one
is interested in algorithm with provable guarantees, one should choose one of the variant \tree over \mfcs or \tgva/\tgvb.
\iflipics{%
\begin{table}[htb]
\resizebox{\textwidth}{!}{
\begin{tabular}{@{}llllllllll@{}}
\multirow{2}{*}{$r$} & \multirow{2}{*}{algorithm} &
\multicolumn{4}{c}{\qquad\quad start with whole $V(G)$} &
\multicolumn{4}{c}{start with $20\%$ of $V(G)$} \\
& & deleted & independent & score & time & deleted & independent & score & time \\
\toprule
\multirow{8}{*}{3} & \mfcs & 5076 & 11471 & 2153 & 0:01.25 & 1922 & 3459 & 1135 & 0:00.48\\
 & \tgva & 78 & 2345 & 2211 & 0:37.53                       & 49 & 1192 & 1159 & 0:29.96\\  
 & \tgvb & 84 & 3820 & 3673 & 0:34.34                      & 49 & 2132 & 2096 & 0:23.36\\  
 & \tgvc & \NA & \NA & \NA & \NA                                   & 5 & 2926 & 2873 & 11:10.63\\  
 & \tree & 7 & 6072 & 5686 & 0:02.77                       & 4 & 2652 & 2598 & 0:00.48\\   
 & \treeshrink & 5 & 5645 & 5645 & 0:01.00                 & 4 & 2603 & 2603 & 0:00.38\\   
 & \ldit & 7 & 6136 & 5748 & 0:01.71                       & 4 & 2741 & 2688 & 0:00.39\\   
 & \ldpow & 5 & 6471 & 6296 & 0:08.13                      & 6 & 2972 & 2871 & 0:02.01\\   
 \midrule
\multirow{8}{*}{5} & \mfcs & 7946 & 15773 & 1164 & 0:01.93& 4057 & 4396 & 594 & 0:00.67\\
 & \tgva & 115 & 1623 & 1445 & 4:38.57                     & 84 & 709 & 676 & 3:20.15\\   
 & \tgvb & 122 & 2079 & 1888 & 4:19.50                    & 103 & 1036 & 982 & 3:07.82\\ 
 & \tgvc & \NA & \NA & \NA & \NA                                  & \NA & \NA & \NA & \NA\\       
 & \tree & 11 & 2988 & 2643 & 0:02.85                     & 4 & 1325 & 1282 & 0:00.53\\  
 & \treeshrink & 5 & 2603 & 2603 & 0:01.05                & 4 & 1284 & 1284 & 0:00.45\\  
 & \ldit & 12 & 3102 & 2752 & 0:01.84                     & 5 & 1380 & 1336 & 0:00.64\\  
 & \ldpow & 7 & 3192 & 3043 & 0:29.32                     & 5 & 1517 & 1473 & 0:07.15\\ 
 \bottomrule\vspace{1pt}
\end{tabular}}
\caption{Aggregated results of uniform quasi-wideness
  on \texttt{medium} set for $r=3$ and $r=5$ (values for $r=2$ and $r=4$ can be found in the full version of the paper):
    the sum of the sizes
  of all deleted and independent sets, 
     total score (sum of all sizes of largest equivalence classes
         w.r.t. deleted vertices),
     and total running time.}\label{tb:uqw}
\end{table}
}{%
\begin{table}[htb]
\resizebox{\textwidth}{!}{
\begin{tabular}{@{}llllllllll@{}}
\multirow{2}{*}{$r$} & \multirow{2}{*}{algorithm} &
\multicolumn{4}{c}{\qquad\quad start with whole $V(G)$} &
\multicolumn{4}{c}{start with $20\%$ of $V(G)$} \\
& & deleted & independent & score & time & deleted & independent & score & time \\
\toprule
\multirow{8}{*}{2} & \mfcs & 587 & 5189 & 4004 & 0:00.88 & 258 & 2426 & 1985 & 0:00.50\\
 & \tgva & 16 & 3892 & 3879 & 0:11.14                     & 6 & 1920 & 1915 & 0:06.12\\
 & \tgvb & 17 & 5990 & 5951 & 0:09.72                    & 22 & 3300 & 3284 & 0:04.64\\
 & \tgvc & 4 & 10013 & 9944 & 39:33.53                   & 2 & 4160 & 4121 & 4:23.58\\
 & \tree & 7 & 8854 & 8424 & 0:02.75                     & 4 & 3834 & 3761 & 0:00.49\\
 & \treeshrink & 5 & 8394 & 8394 & 0:00.98               & 4 & 3770 & 3770 & 0:00.35\\
 & \ldit & 7 & 8985 & 8553 & 0:01.64                     & 4 & 3971 & 3894 & 0:00.40\\
 & \ldpow & 5 & 10169 & 9952 & 0:03.06                   & 3 & 4193 & 4117 & 0:00.89\\
 \midrule
\multirow{8}{*}{3} & \mfcs & 5076 & 11471 & 2153 & 0:01.25 & 1922 & 3459 & 1135 & 0:00.48\\
 & \tgva & 78 & 2345 & 2211 & 0:37.53                       & 49 & 1192 & 1159 & 0:29.96\\  
 & \tgvb & 84 & 3820 & 3673 & 0:34.34                      & 49 & 2132 & 2096 & 0:23.36\\  
 & \tgvc & \NA & \NA & \NA & \NA                                   & 5 & 2926 & 2873 & 11:10.63\\  
 & \tree & 7 & 6072 & 5686 & 0:02.77                       & 4 & 2652 & 2598 & 0:00.48\\   
 & \treeshrink & 5 & 5645 & 5645 & 0:01.00                 & 4 & 2603 & 2603 & 0:00.38\\   
 & \ldit & 7 & 6136 & 5748 & 0:01.71                       & 4 & 2741 & 2688 & 0:00.39\\   
 & \ldpow & 5 & 6471 & 6296 & 0:08.13                      & 6 & 2972 & 2871 & 0:02.01\\   
 \midrule
\multirow{8}{*}{4} & \mfcs & 7269 & 14568 & 1365 & 0:01.73 & 3418 & 4234 & 718 & 0:00.57\\ 
 & \tgva & 106 & 1926 & 1772 & 2:03.13                      & 97 & 886 & 846 & 1:32.35\\    
 & \tgvb & 123 & 2643 & 2471 & 1:53.16                     & 90 & 1361 & 1322 & 1:22.24\\  
 & \tgvc & \NA & \NA & \NA & \NA                                   & \NA & \NA & \NA & \NA\\       
 & \tree & 12 & 3744 & 3388 & 0:02.82                      & 5 & 1726 & 1679 & 0:00.54\\   
 & \treeshrink & 6 & 3344 & 3344 & 0:01.04                 & 5 & 1683 & 1683 & 0:00.39\\   
 & \ldit & 14 & 3959 & 3598 & 0:01.77                      & 5 & 1808 & 1761 & 0:00.56\\   
 & \ldpow & 11 & 4442 & 4079 & 0:20.13                     & 5 & 2004 & 1956 & 0:04.56\\ 
 \midrule
\multirow{8}{*}{5} & \mfcs & 7946 & 15773 & 1164 & 0:01.93& 4057 & 4396 & 594 & 0:00.67\\
 & \tgva & 115 & 1623 & 1445 & 4:38.57                     & 84 & 709 & 676 & 3:20.15\\   
 & \tgvb & 122 & 2079 & 1888 & 4:19.50                    & 103 & 1036 & 982 & 3:07.82\\ 
 & \tgvc & \NA & \NA & \NA & \NA                                  & \NA & \NA & \NA & \NA\\       
 & \tree & 11 & 2988 & 2643 & 0:02.85                     & 4 & 1325 & 1282 & 0:00.53\\  
 & \treeshrink & 5 & 2603 & 2603 & 0:01.05                & 4 & 1284 & 1284 & 0:00.45\\  
 & \ldit & 12 & 3102 & 2752 & 0:01.84                     & 5 & 1380 & 1336 & 0:00.64\\  
 & \ldpow & 7 & 3192 & 3043 & 0:29.32                     & 5 & 1517 & 1473 & 0:07.15\\ 
 \bottomrule\vspace{1pt}
\end{tabular}}
\caption{Aggregated results of uniform quasi-wideness
  on \texttt{medium} set: total size
  of all deleted and independent sets, 
     total score (total size of largest equivalence classes
         w.r.t. deleted vertices),
     and total running time.}\label{tb:uqw}
\end{table}
}

\section{A lower bound to the TGV algorithm}\label{sec:lb}
In this section we observe that the construction of~\cite{GroheKRSS15} shows
also that the bounds of 
our new uniform quasi wideness algorithm of Section~\ref{ss:uqw-tgv} are close to optimal.
More precisely, we show the following corollary of the construction of~\cite{GroheKRSS15}.

\begin{theorem}\label{thm:lb}
For every two integers $k,r \geq 1$ and every integer $m' > c$ where
$c = \binom{k+r}{r}$, there exists a graph $G_{k,r,m'}$ with the following properties:
\begin{itemize}
\item the treewidth of $G_{k,r,m'}$ is at most $k$;
\item $\wcol_r(G_{k,r,m'}) = c$;
\item $|V(G_{k,r,m'})| \geq (m'-1)^c$.
\item for every pair of disjoint sets $B,Z \subseteq V(G_{k,r,m'})$ such that $B$ is $2r$-independent
in $G_{k,r,m'}-Z$, we have $|B| \leq |Z| \cdot m' + 1$; in particular, if $|B| \geq cm' + 1$ then
$|Z| \geq c$ and if $|Z| \leq c$ then $|B| \leq cm' + 1$.
\end{itemize}
\end{theorem}
Before we proceed with the proof, let us discuss the statement and its implications.
Most importantly,
the example of Theorem~\ref{thm:lb} is weak in the sense that it treats $2r$-independent
sets, as opposed to $r$-independent sets output by the algorithm of Section~\ref{ss:uqw-tgv}.
However, it shows that even in bounded treewidth graph classes the dependency between
the size of the input set $A$ and the size of the output independent set $B$
needs to be polynomial with degree depending on the quality of the graph class in question
(here, $c = \wcol_r(G_{k,r,m'})$).
Apart from this slackness, the bounds in Theorem~\ref{thm:lb} are very similar to the ones
of Theorem~\ref{thm:uqw-tgv}: to get an independent set of size $m := cm'+1$ in a graph
with $\wcol_r = c$ one needs
a vertex set of a graph of size $(m'-1)^c \sim (m/c)^c$ 
and the deletion of $c$ vertices.

\begin{proof}[Proof of Theorem~\ref{thm:lb}.]
We start by recalling the construction of~\cite{GroheKRSS15}.
Fix a branching degree $d$.
For every $k,r \geq 1$ let $T(k,r)$ be a rooted tree of depth $c = \binom{k+r}{r}$
and branching degree $d$.
We define graphs $G(k,r)$ inductively as follows. 

First, we start with $T(k,r)$ being a spanning tree of $G(k,r)$.
We will maintain the invariant that every edge of $G(k,r)$ connects an ancestor and a descendant
in $T(k,r)$ (i.e., $G(k,r)$ is a subgraph of ancestor-descendant closure of $T(k,r)$).

For $k=1$, we take $G(k,r) = T(k,r)$. 
For $r=1$, we take $G(k,r)$ to be the whole ancestor-descendant closure of $T(k,r)$, that is,
we add $uv$ to $E(G(k,r))$ whenever $u$ is an ancestor of $v$ in $T(k,r)$.
For $k,r \geq 2$, note that one can equivalently construct $T(k,r)$ as follows:
start with $T(k,r-1)$ and for every leaf $v$ of $T(k,r-1)$, create $d$ copies of $T(k-1,r)$
and connect their roots to $v$.
To define $G(k,r)$, we proceed as follows: we start with $G(k,r-1)$ and for every leaf
$v$ of the spanning tree $T(k,r-1)$ of $G(k,r-1)$, we create $d$ copies of $G(k-1,r)$
and make all of them fully adjacent to $v$.

In~\cite{GroheKRSS15}, it is shown that the treewidth
of $G(k,r)$ is $k$, and that as long as $d \geq c = \binom{k+r}{r}$, in every
ordering~$L$ of $V(G(k,r))$ there exists a leaf $v$ of $T(k,r)$ with its every ancestor
belonging to $\operatorname{WReach}_r[G(k,r), L, v]$ (in particular, $\wcol_r(G(k,r)) \geq c$).
We take $G(k,r,m') = G(k,r)$ for branching degree $d = m'-1$; recall that $m' > c = \binom{k+r}{r}$.
The bound on the number of vertices of $G(k,r,m')$ is straightforward.
It remains to show the last property of $G(k,r,m')$.

We start by observing the following.
\begin{claim}\label{cl:lb:jump}
For every $v \in V(G(k,r))$ and its ancestor $u$ in $T(k,r)$, there exists
a path from $v$ to $u$ of length at most $r$ that traverses only vertices on the unique
path from $v$ to $u$ in $T(k,r)$.
\end{claim}
\begin{proof}
We proceed by induction on $k+r$. For $k=1$ or $r=1$ the statement is straightforward.
Assume then $k,r \geq 2$, and recall that $G(k,r)$ consists of $G(k,r-1)$ and
$d$ copies of $G(k-1,r)$ attached to every leaf of $T(k,r-1)$.

If $u$ and $v$ both belong to $G(k,r-1)$ or to the same copy of $G(k-1,r)$, then we are done
by the inductive hypothesis. Otherwise, $v$ belongs to a copy of $G(k-1,r)$ attached to a leaf $w$
of $T(k,r-1)$, and $u$ belongs to $G(k,r-1)$. By the inductive hypothesis, there exists
a path of length at most $r-1$ from $w$ to $u$ that uses only vertices on the path from
$w$ to $u$ in $T(k,r-1)$. Together with the edge $vw$, this path forms the desired path
from $v$ to $u$.
\cqed\end{proof}
Consequently, for every subtree $T$ of $T(k,r)$, every vertex of $T$
is within distance at most $r$ from the topmost vertex of $T$ in $G(k,r)$, and, consequently,
the vertex set of $T$ induces a graph of diameter at most $2r$ in $G(k,r)$.

Consider now a pair of disjoint sets $B,Z \subseteq V(G)$ such that $B$ is $2r$-independent
in $G-Z$. The observation from the preceding paragraph implies that every connected
component of $T(k,r)-Z$ contains at most one vertex of $B$. On the other hand, the maximum
degree of $T(k,r)$ is $d+1 = m'$. Consequently, $|B| \leq |Z|m' + 1$. This finishes
the proof.
\end{proof}


\section{Conclusions}

We have conducted a thorough empirical evaluation of algorithms for computing generalized coloring numbers and uniform quasi-wideness.
In the case of the weak coloring number, one of the simplest heuristics achieved very good results and was
only outperformed by two greedy heuristics that also do not enjoy any theoretical guarantees.
For uniform quasi-wideness, again the simplest heuristic outperformed all other approaches.
From the algorithms with provable guarantees, the experiments indicated
a variant of the algorithm of~\cite{HeuvelMQRS17} as the algorithm
of choice for generalized coloring numbers and
a variant of the algorithm of~\cite{Pil2017number} as the algorithm of choice
for uniform quasi-wideness.

Furthermore, 
our new algorithm for uniform quasi-wideness, whose development was motivated by the conservativeness of the previous approach of~\cite{KreutzerPRS16}, performed rather poorly in the experiments. Our explanation for this result 
is that the main combinatorial idea in this approach, to restrict the search space upon deletion step to the weakly reachable set
of the deleted vertex, while necessary for the theoretical guarantee on the size of the deletion set, is too conservative in practice.

As a direction for future work, we would like to suggest a more in-depth study of the distribution 
of the values of generalized coloring numbers in different classes of real-world networks, similarly as it is 
done for $p$-treedepth colorings in~\cite{ComplexNetworks}.
Furthermore, as discussed in Section~\ref{sec:isnaive}, one could explore the possibility of using 
more sophisticated maximum independent set heuristics to improve upon the simplest heuristic for uniform quasi-wideness. 
It would also be interesting to find and implement efficient heuristics for lower bounds of weak coloring numbers. A small gap between them and our upper bounds would mean that both have a good quality. Otherwise we would know that there is room for improvement.  
Finally, it would be interesting to use the findings of this work for some start-to-end pipeline for a problem such as
motif counting (see~\cite{OBrienSullivan} for experimental evaluation of a pipeline using $p$-treedepth colorings).



\section*{Acknowledgments}

We thank Christoph Dittmann for providing us with his code for the
\mfcs algorithm, which we partially used for our implementation. 
We thank Micha\l{} Pilipczuk for many hours of
fruitful discussions.
Furthermore, we thank anonymous reviewers for their very valuable and in-depth comments on the manuscript.

\bibliographystyle{abbrv}
\bibliography{refs-corrected}

\end{document}

\end{document}
